\documentclass[final]{siamltex}
\usepackage{amssymb}
\usepackage{graphicx}
\usepackage{enumerate}
\usepackage{amsmath}
\input pictex.sty

\newcommand{\pd}{\partial}
\newcommand{\no}{\nonumber}

\newcounter{space}

\begin{document}

\title{On the Whitham Equations for the Defocusing Complex 
Modified KdV Equation\thanks{Department of Mathematics, Ohio State
  University, 231 W. 18th Avenue} }%
\author{Yuji Kodama\thanks{{\tt kodama@math.ohio-state.edu}}
\and V. U. Pierce\thanks{ {\tt vpierce@math.ohio-state.edu}}
\and Fei-Ran Tian\thanks{
 {\tt tian@math.ohio-state.edu} 
}}
\maketitle

\begin{abstract}
We study the Whitham equations for the
defocusing complex modified KdV (mKdV) equation. 
These Whitham equations are quasilinear hyperbolic equations
and they describe the averaged dynamics of the rapid oscillations which
appear in the solution of the mKdV equation when the dispersive parameter is small.
The oscillations are referred to as dispersive shocks.
The Whitham equations for the mKdV equation are neither strictly hyperbolic nor
genuinely nonlinear. We are interested in the solutions of the Whitham equations
when the initial values are given by a step function. 
We also compare the results with those of the defocusing nonlinear Schr\"odinger (NLS) equation.
For the NLS equation, the Whitham equations are strictly hyperbolic and genuinely nonlinear.
We show that the weak hyperbolicity of the mKdV-Whitham equations is responsible
for an 
additional structure in the dispersive
shocks which has not been found in the NLS case.

\end{abstract}


\markboth{ Y. KODAMA, V.U. PIERCE AND F.-R. TIAN}{
WHITHAM EQUATIONS FOR THE COMPLEX MODIFIED KDV}
\pagestyle{myheadings}
\thispagestyle{plain}

\section{Introduction}

In \cite{PT, PT2}, Pierce and Tian studied the self-similar solutions of
the Whitham equations which describe the zero dispersion limits of the KdV hierarchy. 
The main feature of the Whitham
equations for the higher members of the hierarchy, of which the KdV equation is the
first member, is that these Whitham equations are neither strictly hyperbolic
nor genuinely nonlinear. This is in sharp contrast to the case of the KdV equation
whose Whitham equations are strictly hyperbolic and genuinely nonlinear \label{lax2}. In this paper,
we extend their studies to the case of the complex modified KdV equation,
which is the second member of the defocusing nonlinear Schr\"odinger hierarchy.
The Whitham equations for the defocusing NLS equation are strictly hyperbolic
and genuinely nonlinear, and they have been studied extensively (see for examples,
\cite{for, jin, kod, pav, Tian1}). However, for the mKdV equation, the Whitham equations are neither 
strictly hyperbolic nor genuinely nonlinear.

Let us begin with a brief description of 
the zero dispersion limit of the solution of the NLS equation
\begin{equation}
\label{nls}
\sqrt{-1} \ \epsilon \frac{\partial \psi}{\partial t} +
2 \epsilon^2 \frac{\partial^2
\psi}{\partial x^2} - 4 |\psi |^2 \psi= 0
\ ,
\end{equation}
with the initial data
\begin{equation*}
\psi (x, 0) = A_0(x)
\exp\left(\sqrt{-1} \ \frac{S_0(x)}
{\epsilon}\right) \ .
\end{equation*}
Here $A_0(x)$
and $S_0(x)$ are real functions that are independent
of $\epsilon$.
Writing the solution
$\psi (x,t;\epsilon)=A(x,t;\epsilon)\exp\left(\sqrt{-1} \ \frac{S(x,t;\epsilon)}{\epsilon}
\right)$, and using the notation $\rho(x,t;\epsilon)=A^2(x,t;\epsilon)$, $v(x,t;\epsilon) = {\partial S(x,t,\epsilon)
/ \partial x}$,
one obtains the conservation form of the defocusing NLS equation
\begin{equation}
\label{nlshydro}
\left\{
\begin{array}{lllll}
&\displaystyle{\frac{\partial \rho}{\partial
t} + \frac{\partial}{\partial
x}(4 \rho v)} = 0 \
,\\
&\displaystyle{\frac{\partial}{\partial
t}(\rho v) + \frac{\partial}{\partial
x}
\left(4 \rho v^2+2 \rho^2 \right) = \epsilon^2
\frac{\partial}{\partial
x}\left(\rho\frac{\partial^2}{\partial
x^2}\ln \rho
\right)} \ .
\end{array}
\right.
\end{equation}

The mass density $\rho=|\psi|^2$ and momentum density $\rho v=\frac{\sqrt{-1}}{2}(\psi\psi^*_x-\psi^*\psi_x)$
have weak limits as $\epsilon \rightarrow 0$ \cite{jin}. 
These limits satisfy a $2 \times 2$ system of hyperbolic equations
\begin{equation}
\label{nlshydro0}
\left\{
\begin{array}{lllll}
&\displaystyle{\frac{\partial \rho}{\partial
t} + \frac{\partial}{\partial
x}(4 \rho v)} = 0 \
,\\
&\displaystyle{\frac{\partial}{\partial
t}(\rho v) + \frac{\partial}{\partial
x}
\left(4 \rho v^2+2\rho^2 \right) = 0}
\ ,
\end{array}
\right.
\end{equation}
until its solution develops a shock. 
System (\ref{nlshydro0}) can be rewritten in the diagonal form for $\rho\ne 0$,
\begin{equation}
\label{nls0}
\frac{\partial}{\partial
t}\begin{pmatrix} {\alpha}\\
{\beta}\end{pmatrix} +2
\begin{pmatrix}
3{\alpha}+{\beta} & 0 \\ 0 &
{\alpha}+3{\beta}\end{pmatrix}
\frac{\partial}{\partial
x}\begin{pmatrix} {\alpha} \\
{\beta}\end{pmatrix} = 0 \ ,
\end{equation}
 where the
Riemann invariants ${\alpha}$ and ${\beta}$ are
given by
\begin{equation}
\label{NSFII}
\alpha ={v \over 2} + \sqrt{\rho}\ , \quad \beta ={v \over 2} - \sqrt{\rho} \ .
\end{equation}

As a simple example, we consider the case with $\alpha=a=$constant. System (\ref{nls0})
reduces to a single equation
\begin{equation}\label{nlsbeta}
\frac{\partial\beta}{\partial t}+2(a+3\beta)\frac{\partial\beta}{\partial x}=0\,.
\end{equation}
The solution is given by the implicit form
\[
\beta(x,t)=f(x-2(a+3\beta)t) \ ,
\]
where $f(x) = \beta(x,0)$ is the initial data for $\beta$.
One can easily see that if $\beta(x,0)$ decreases in some region, then $\beta(x,t)$
develops a shock in a finite time, i.e., $\partial\beta/\partial x$ becomes singular.

After the shock formation in the solution of (\ref{nlshydro0}) or (\ref{nls0}), the weak
limits are described by the NLS-Whitham equations, which can also be put in 
the Riemann invariant form \cite{for, jin, kod, pav}
\begin{equation}
\label{nlsw}
\frac{\partial u_i}{\partial t} + \lambda_{g,i}(u_1,\ldots,u_{2g+2})\frac{\partial u_i}{\partial x} = 0 \ , \quad i = 1, 2, \ldots, {2g+2} \ ,
\end{equation}
where $\lambda_{g,i}$ are expressed in terms of complete hyperelliptic  integrals 
of genus $g$ \cite{lax2}. Here the number
 $g$ is exactly the number of phases in the NLS oscillations with
small dispersion.
Accordingly, the zero phase $g=0$ corresponds to no oscillations, and single and higher phases
$g \geq 1$ correspond to the NLS oscillations. System (\ref{nls0}) is viewed as the zero phase
Whitham equations. 
The solution of the Whitham equations (\ref{nlsw}) for $g \geq 1$ then describes
the averaged motion of the oscillations appearing in the solution of (\ref{nls}) (see e.g. \cite{kod}).

Let us discuss the most important  $g=1$ case in more detail. We note that it is well known that 
the KdV oscillatory solution, in the single phase regime, can be approximately described by the KdV
periodic solution when the dispersive parameter is small \cite{dei, gra, ven}. It is very possible to use
the method of \cite{dei, ven} to show that the solution of the NLS equation
(\ref{nls}) for small $\epsilon$ can be approximately described, in the single phase regime, by
 the periodic solution of the NLS equation. The NLS periodic solution has the form
\begin{equation}\label{periodicsolution}
\tilde{\rho}(x,t;\epsilon)= \rho_3+(\rho_2-\rho_3)\,{\rm sn}^2(\sqrt{\rho_1-\rho_3}\,\theta(x,t;\epsilon), s)\,.
\end{equation}
with $\theta(x,t;\epsilon)=(x-V_1t)/\epsilon$ and the velocity $V_1=V_1(\rho_1,\rho_2,\rho_3)$.
Here $\rho_i$'s are determined by the equation obtained from (\ref{nlshydro})
\[
\frac{\epsilon^2}{4}\left(\frac{d\rho}{d\theta}\right)^2=(\rho-\rho_1)(\rho-\rho_2)(\rho-\rho_3)\,
\]
with $\rho_1>\rho_2>\rho_3$, and ${\rm sn}(z,s)$ is the Jacobi elliptic function with
the modulus $s=(\rho_2-\rho_3)/(\rho_1-\rho_3)$.  
We can also write $\rho_i$'s as
\cite{el}
\begin{equation}\label{rhos}\left\{\begin{array}{lll}
&\displaystyle{\rho_1 = \frac{1}{4}(u_1+u_2-u_3-u_4)^2}\\
&\displaystyle{\rho_2 =\frac{1}{4}(u_1-u_2+u_3-u_4)^2}\\
&\displaystyle{\rho_3 =\frac{1}{4}(u_1-u_2-u_3+u_4)^2}
\end{array}\right.
\end{equation}
with $u_1>u_2>u_3>u_4$.   The velocity $V_1$ is then given by
\begin{equation*}
V_1=2(u_1+u_2+u_3+u_4)\,.
\end{equation*}
For constants  $u_1$, $u_2$, $u_3$ and $u_4$, formula (\ref{periodicsolution}) gives the well known elliptic
solution of the NLS equation. To describe the solution $\rho(x,t;\epsilon)$ of  the NLS equation (\ref{nlshydro}), the quantities
$u_1$, $u_2$, $u_3$ and $u_4$ are instead functions of $x$ and $t$ and they evolve according to the single phase Whitham equations (\ref{nlsw}) for $g=1$.

The weak limit of $\rho(x,t;\epsilon)$ of NLS equation (\ref{nls}) as 
$\epsilon \rightarrow 0$ can be expressed in terms of 
 $\rho_1$, $\rho_2$, $\rho_3$ and $\rho_4$ \cite{jin}
 \begin{equation}\label{average}
 \overline{\rho(x,t)}=\rho_1-(\rho_1-\rho_3)\frac{E(s)}{K(s)}\,,
\end{equation}
where $K(s)$ and $E(s)$ are the complete elliptic integrals of the first and second kind, respectively. This weak limit can also be viewed
as the average value of the periodic solution $\tilde{\rho}(x,t;\epsilon)$ of (\ref{periodicsolution}) over its period $L= 2 \epsilon K(s)/ \sqrt{\rho_1-\rho_3}$.

In order to see how a single phase Whitham solution appears, we consider
the following step initial data for system (\ref{nls0})
\begin{equation}\label{nlsinitial}
\alpha(x,0)=a,\quad \beta(x,0)=\begin{cases}b,~~ x<0 \\ c,~~ x>0\end{cases}\,,
\end{equation}
where $a > b$, $a > c$, $b \neq c$. The solution of (\ref{nls0}) develops a shock if and only if
$b > c$ (cf. (\ref{nlsbeta})). After the formation of a shock, the Whitham equations (\ref{nlsw}) with $g=1$
kick in. For instance,  we consider the Whitham equations with the initial data \cite{kod}
\begin{equation}\label{regular}
u_1(x,0)=a, \quad u_2(x,0)=b,\quad u_3(x,0)=\begin{cases}b,~~x<0\\ c,~~x>0\end{cases},\quad
u_4(x,0)=c.
\end{equation}
Now notice that the Whitham equations (\ref{nlsw}) for $g=1$ with the initial data (\ref{regular}) can be reduced to a single equation 
$u_{3t}+\lambda_{1,3}(a,b,u_3,c)u_{3x}=0$. 
The equation has a global self-similar
solution, which is implicitly given by
$x/t=\lambda_3(a,b,u_3,c)$.
The $x$-$t$ plane is then divided into {\it three} parts
$$(1) \   \frac{x}{t} <  \gamma \ , \quad
(2) \  \gamma < \frac{x}{t} < 2a + 4b + 2c \ , \quad
(3) \  \frac{x}{t} > 2a + 4b + 2c \ ,$$ 
where $\gamma = 2(a + b + 2c) - 8(a-c)(b-c)/(a+b - 2c)$ (see (\ref{bc3}) and (\ref{bc4}) below
for the derivation).
The solution of system (\ref{nls0}) occupies the first and third parts, i.e.,
\begin{itemize}
\item[(1)] for $x/t<\gamma$,
\[
{\alpha}(x, t) = a,\quad {\beta}(x, t)= b,
\]
\item[(3)] for $x/t>2a+4b+2c$,
\[
{\alpha}(x, t) = a,\quad {\beta}(x,t)= c \ .
\]
\end{itemize}
The Whitham solution of (\ref{nlsw}) with $g=1$ lives in the second part, i.e., 
\begin{itemize}
\item[(2)] for $\gamma<x/t<2a+4b+2c$, 
\[
u_1(x, t) = a \ , \quad u_2(x, t) = b \ , \quad
\frac{x}{t} = \lambda_{1,3}(a, b, u_3, c) \ , \quad  u_4(x, t) = c
\ , 
\]
\end{itemize}
where the solution $u_3$ can be obtained as a function of the
self-similarity variable $x/t$, if 
\begin{equation*}  
\frac{\pd \lambda_{1,3}}{\pd u_3} (a, b, u_3, c) \neq 0. 
\end{equation*}
Indeed, it has been shown that the Whitham equations (\ref{nlsw}) are
genuinely nonlinear \cite{jin, kod}, i.e.,
\begin{equation}\label{eq11}
\frac{\pd \lambda_{1,i}}{\pd u_i} (u_1, u_2, u_3, u_4) > 0, \,
\quad i=1, 2, 3, 4, 
\end{equation}
for $u_1 > u_2 > u_3 > u_4$. In Figure \ref{fig0}, we plot the self-similar solution
of the Whitham equations (\ref{nlsw}) with $g=1$  for the NLS equation, and the corresponding periodic oscillatory solution
(\ref{periodicsolution}) for the initial data (\ref{nlsinitial}) with $a=4,b=0$ and $c=-1$.
The oscillations describe a dispersive shock of the NLS equation under a small dispersion. Note here that the oscillations have a uniform structure, which is due to an 
almost linear
profile of the Whitham solution $u_3$. This will be seen to be in sharp contrast 
to the  case of
the mKdV equation, which we will discuss later (cf. Figure \ref{fig1}).

\begin{figure}[h] \label{fig0}
\begin{center}
\includegraphics[width=12cm]{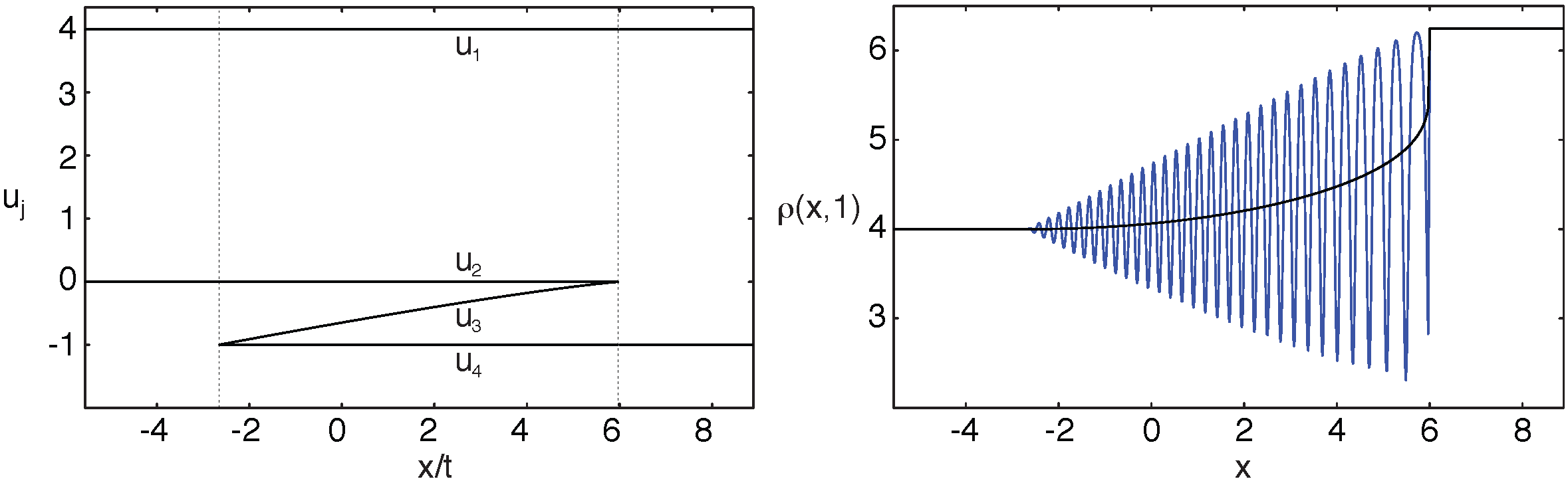}
\caption{Self-Similar solution of the NLS-Whitham equation (\ref{nlsw}) of $g=1$ and the
 corresponding oscillatory solution (\ref{periodicsolution}) of the NLS equation with
 $\epsilon = 0.15$. The dark line in the middle of the oscillations is 
the weak limit
$\overline{\rho(x,t)}$ given by (\ref{average}).
 The initial data is given by (\ref{nlsinitial}) with
$a=4$, $b=0$ and $c=-1$.}
\end{center}\end{figure}

The figures in this paper all have the same form:  On the left hand
side is a plot of the solution of the Whitham equations as a function of the
self-similarity variable $x/t$, which is
exact, other than a numerical method used to implement the inverse
function theorem.  
On the right hand side is the
oscillatory solution given by (\ref{periodicsolution}) (respectively
(\ref{periodicsolution2}) for mKdV) at $t=1$, while the
dark plot  is the weak limit (\ref{average})
 of the oscillatory solution at $t=1$,
both
plots on the right are also exact.  
In the first two figures we demark the region where
the Whitham equations with $g=1$ govern the solution, 
place a dashed-dotted line where the behavior of the oscillatory
solution changes, 
 and label the
four functions $u_1 > u_2 > u_3 > u_4$.  The demarcation and
labeling are similar in the other figures and we will leave them off
for brevity.  
Although we do not include any numerical simulations, we would like to mention that
E. Overman showed us his numerical simulation of the NLS and mKdV
equations
which captures the features of the oscillatory solutions plotted in
this paper.

The defocusing NLS equation is just the first member of the defocusing NLS hierarchy;
the second is the (defocusing) complex modified KdV (mKdV) equation
\begin{equation}
\label{mkdv}
{\partial \psi \over \partial t} + {3 \over 2} \ |\psi|^2 {\partial \psi \over \partial x} - {\epsilon^2 \over 4}
{\partial^3 \psi \over \partial x^3} = 0 \ .
\end{equation}
We again use
$\psi (x,t;\epsilon)=A(x,t;\epsilon)\exp\left(\sqrt{-1} \ \frac{S(x,t;\epsilon)}{\epsilon}
\right)$ and notation $\rho(x,t;\epsilon) = A^2(x,t;\epsilon)$, $v(x,t;\epsilon) = {\partial S(x,t;\epsilon)
/ \partial x}$
to obtain the conservation form of the mKdV equation
\begin{equation}
\left\{
\begin{array}{lllll}
&\displaystyle{\frac{\partial \rho}{\partial
t}  + \frac{\partial}{\partial
x} \left( {3 \over 4} \rho^2 + {3 \over 4} \rho v^2 \right) = {\epsilon^2}
{\partial \over \partial x} \left( \rho^{3/4}{\partial^2  \over \partial x^2} \rho^{1/4}\right)} \ , \\
&\displaystyle{\frac{\partial}{\partial
t}(\rho v) + \frac{\partial}{\partial
x}
\left( {3 \over 2} \rho^2 v + {3 \over 4} \rho v^3 \right) = \frac{\epsilon^2}{4}
\ \frac{\partial}{\partial
x}\left[ \frac{\partial^2}{\partial x^2} (\rho v) - {3 \over 2} R 
\right]} \ ,
\end{array}
\right.
\end{equation}
where
$$R =  \frac{3v}{2\rho}\left(\frac{\partial \rho}{\partial x}\right)^2+\frac{\partial v}{\partial x}\frac{\partial \rho}{\partial x}-v\frac{\partial^2\rho}{\partial x^2}\,.$$

The mass density $\rho$ and momentum density $\rho v$ for the mKdV equation
also have weak limits as $\epsilon \rightarrow 0$ \cite{jin}. As in the NLS case, the weak limits
satisfy 
\begin{equation}
\label{mkdvhydro0}
\left\{
\begin{array}{lllll}
&\displaystyle{\frac{\partial}{\partial
t} \rho + \frac{\partial}{\partial
x} \left( {3 \over 4} \rho^2 + {3 \over 4} \rho v^2 \right) = 0} \ , \\
&\displaystyle{\frac{\partial}{\partial
t}(\rho v) + \frac{\partial}{\partial
x}
\left( {3 \over 2} \rho^2 v + {3 \over 4} \rho v^3 \right) = 0} \ ,
\end{array}
\right.
\end{equation}
until the solution of (\ref{mkdvhydro0}) forms a shock. 
One can rewrite equations (\ref{mkdvhydro0}) as
\begin{equation}
\label{mkdv0}
\frac{\partial}{\partial
t}\begin{pmatrix} \alpha\\
\beta\end{pmatrix} +
\frac{3}{8}
\begin{pmatrix}
5 \alpha^2 + 2 \alpha \beta + \beta^2 & 0 \\ 0 &
\alpha^2 + 2 \alpha \beta + 5 \beta^2 \end{pmatrix}
\frac{\partial}{\partial
x}\begin{pmatrix} \alpha \\
\beta\end{pmatrix} = 0 \ ,
\end{equation}
 where the
Riemann invariants $\alpha$ and $\beta$ are
again given by formula (\ref{NSFII}).

Let us again consider the simplest case $\alpha(x,0)=a$, where a is
constant, to see how 
the solution of system (\ref{mkdv0}) develops a shock. In this case, the system reduces to a single equation
\begin{equation*}
\frac{\partial \beta}{\partial t}+\frac{3}{8}\left(a^2+2a\beta+5\beta^2\right)\frac{\partial \beta}{\partial x}=0\,.
\end{equation*}
As in the NLS case, we consider the initial function given by
 $\beta(x,0)=b$ for $x<0$ and $\beta(x,0)=c$
for $x>0$. We recall that, in the NLS case, the zero phase solution of (\ref{nls0})
develops a shock if and only if $b > c$.
However, the solution in the 
mKdV case develops a shock  for $b>c$ if and only if $a + 5b >0$.
In addition, if $b<c$,  the solution in the mKdV case develop a
 shock if and only if $a + 5b < 0$.
These differences between the mKdV and NLS cases 
are due to the weak hyperbolicity of the system (\ref{mkdv0}) (note that
for the eigenspeed $\lambda=\frac{3}{8}(\alpha^2+2\alpha\beta+5\beta^2)$ for $\beta$,
we have $\partial \lambda/\partial \beta=\frac{3}{4}(\alpha+5\beta)$ which can change sign).
As will be shown below, this leads to an additional structure in the dispersive shock for the mKdV case.

As in the case of the NLS equation, immediately after the shock formation in the solution of (\ref{mkdvhydro0}), the weak
limits are described by the mKdV-Whitham equations
\begin{equation}
\label{mkdvw}
\frac{\partial u_i}{\partial t} + \mu_{g,i}(u_1, \ldots, u_{2g+2})\frac{\partial u_i}{\partial x} = 0 \ , \quad i = 1, 2, 3\ldots, 2g+2 \  ,
\end{equation}
where $\mu_{g,i}$'s can also be expressed in terms of complete hyperelliptic integrals 
of genus $g$ \cite{jin}.

In this paper, we study the solution of the Whitham equations (\ref{mkdvw}) with $g=1$ when the 
initial mass density $\rho(x,0)$ and momentum density $\rho(x,0)v(x, 0)$ are
step functions. In view of (\ref{NSFII}), this amounts to requiring $\alpha$ and $\beta$ of system 
(\ref{mkdv0}) to
have step-like initial data. We are interested in the following two cases:
\begin{itemize}
\item[(i)]  $\alpha(x,0)$ is a constant and 
\begin{equation} \label{step1}
\alpha(x, 0) = a \ , \quad \beta(x,0) = \left\{ \begin{matrix} b, & x < 0 \\
c, & x > 0 \end{matrix} \right. \ , \ \quad a > b \ , a > c \ , b \neq c \ ,
\end{equation}
\item[(ii)]
$\beta(x, 0)$ is a constant and
\begin{equation} \label{step2}
\quad \alpha(x,0) = \left\{ \begin{matrix} b, & x < 0 \\
c, & x > 0 \end{matrix} \right. \ , \quad \beta(x, 0) = a \ , \quad b > a \ , c > a \ , b \neq c \ .
\end{equation}
\end{itemize}

In the case of the NLS equation, the genuine nonlinearity of the single phase 
Whitham equations (see (\ref{eq11}))
warrants that the solution is found
 by the implicit function theorem.
However the mKdV-Whitham equations (\ref{mkdvw}), in
general, are not genuinely nonlinear, that is, a property like
(\ref{eq11}) is not available (see Lemma \ref{lem2.1} below).  
Our construction of solutions of the Whitham
equation (\ref{mkdvw}) with $g=1$ makes use of the non-strict
hyperbolicity of the equations.  For the NLS case, it has been known in \cite{jin, kod} that the Whitham
equations (\ref{nlsw}) with $g=1$ are strictly hyperbolic, that is,
\begin{equation*}
\lambda_{1,1} > \lambda_{1,2} > 
\lambda_{1,3} > \lambda_{1,4}  
\end{equation*}
for $u_1 > u_2 > u_3 > u_4$.  For the mKdV-Whitham equations (\ref{mkdvw}) with $g=1$, the eigenspeeds
$\mu_{1,i}(u_1, u_2, u_3, u_4)$ may coalesce in the region $u_1 > u_2 > u_3 > u_4$ (we will discuss
the details in Section \ref{whithamsec}).

Let us now describe one of our main results (see Theorem \ref{thm1}) 
for the single phase mKdV-Whitham equations with step-like initial 
function (\ref{step1}) for $a=4$, $b=1$ and $c= -1$. In this case,
the space time is 
divided into {\it four} regions (see Figure \ref{fig1}) instead of {\it three} in the case
of the NLS equation (cf. Figure \ref{fig0})
$$(1) \ \frac{x}{t} < c_1 \ , \quad
(2) \ c_1 < \frac{x}{t} < c_2 \ , \quad
(3) \ c_2 < \frac{x}{t} < c_3 \ , \quad
(4) \ \frac{x}{t} > c_3 \ ,$$
where $c_1$, $c_2$ and $c_3$ are some constants. 
In the first and fourth regions, the solution of the $2\times 2$ system (\ref{mkdv0}) governs
the evolution:
\begin{itemize}
\item[(1)] for $x/t<c_1$,
\[
\alpha(x,t) = 4,\quad \beta(x,t)= 1,
\]
\item[(4)]  for $x/t>c_2$,
\[
\alpha(x,t)=4,\quad  \beta(x,t)=-1 .
\]
\end{itemize}
The Whitham solution of the $4\times 4$ system (\ref{mkdvw}) with $g=1$ lives in the second and third
regions; 
\begin{itemize}
\item[(2)] for $c_1<x/t<c_2$,
\begin{equation*} 
\label{ns}
u_1(x, t) = 4, \quad u_2(x, t)=1, \quad   
\frac{x}{t} = \mu_{1,3}(4, 1, u_3, u_4) , \quad
\frac{x}{t} = \mu_{1,4}(4, 1, u_3, u_4) \ , 
\end{equation*}
\item[(3)] for $c_2<x/t<c_3$,
\[
u_1(x, t)= 4 , \quad u_2(x, t)= 1 , \quad 
\frac{x}{t} = \mu_{1,3}(4, 1, u_3, -1) \ , \quad  u_4(x, t)=-1 .
\]
\end{itemize}
Note that, in the second region, we have
$$ \mu_{1,3}(4, 1, u_3, u_4) = \mu_{1,4}(4, 1, u_3, u_4)$$
on a curve in the region $-1 < u_4 < u_3 < 4$.
This implies the non-strict hyperbolicity of the mKdV-Whitham equations (\ref{mkdvw}) for $g=1$.

It is again possible to use the method of \cite{dei, ven} to show that the solution of
the mKdV equation 
(\ref{mkdv}) can be approximately described, in the single phase regime,
by the periodic solution of the mKdV when $\epsilon$ is small. The periodic solution has the same form as 
(\ref{periodicsolution})
of the NLS, i.e.,
\begin{equation}\label{periodicsolution2}
\tilde{\rho}(x,t;\epsilon)= \rho_3+(\rho_2-\rho_3)\,{\rm sn}^2(\sqrt{\rho_1-\rho_3}\,\theta(x,t;\epsilon), s)\, .
\end{equation}
However, $\theta(x,t;\epsilon)$ is now given by $\theta=(x-V_2t)/\epsilon$ with the velocity $V_2$ (see e.g. \cite{bel})
\[
V_2=\frac{3}{8}\sigma_1^2-\frac{1}{2}\sigma_2\,,
\]
where $\sigma_1=\sum_{j=1}^4u_j$ and $\sigma_2=\sum_{i<j}u_iu_j$ are
the elementary symmetric functions of degree one and two, respectively.
The functions $\rho_1, \rho_2$ and $\rho_3$ are also given by formula (\ref{rhos}).
If  
$u_1$, $u_2$, $u_3$ and $u_4$ are constants, formula (\ref{periodicsolution2}) gives the periodic solution of
the mKdV equation. To describe the solution $\rho(x,t;\epsilon)$ of the mKdV equation
(\ref{mkdv}), the quantities $u_1$, $u_2$, $u_3$ and $u_4$ must satisfy the 
single phase mKdV-Whitham
equations (\ref{mkdvw}) for $g=1$.  The weak limit of $\rho(x,t;\epsilon)$ of
the mKdV equation is also given by formula (\ref{average}). 

In Figure \ref{fig1}, we plot the self-similar solution of the Whitham equations (\ref{mkdvw}) for $g=1$
and the corresponding periodic oscillatory solution (\ref{periodicsolution2}).  We note here that
the pattern of the oscillation in this case has two distinct structures: one corresponds to the region
(2), $c_1<x/t<c_2$, and the other corresponds to the region (3), $c_2<x/t<c_3$, We also note that the weak limit
 $\overline{\rho(x,t)}$ is not $C^1$ smooth at the boundary point $x/t=c_2\approx 4.63$.
 
\begin{figure}[h] 
\begin{center}
\includegraphics[width=12cm]{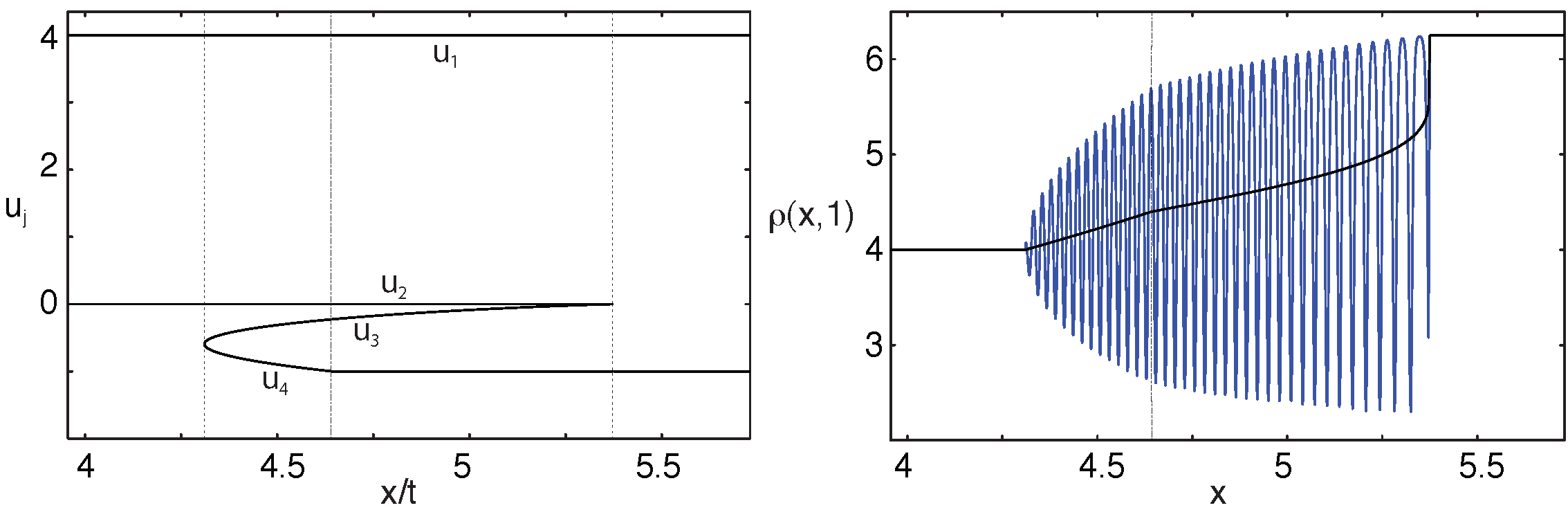}
\caption{\label{fig1} Self-Similar solution of the mKdV-Whitham equation (\ref{mkdvw}) with $g=1$ and the corresponding
oscillatory solution (\ref{periodicsolution2}) of the mKdV equation with $\epsilon = 0.012$. 
The initial data is given by (\ref{step1}) with
  $a=4$, $b=0$ and $c=-1$. There are two distinct structures in the oscillations and they are
  separated by $x/t\approx 4.63$ (cf. Figure \ref{fig0}).}
\end{center}\end{figure}

As we will show below, for other values of $a$, $b$ and $c$,  the solutions of (\ref{mkdv0}) and (\ref{mkdvw}) with $g=1$
will be 
seen to be quite different from
the above.

The Whitham equations (\ref{mkdvw}) for the mKdV equation are analogous to the Whitham equations for the
higher members of the KdV hierarchy \cite{PT, PT2}.  For step like initial
data, the single phase Whitham solutions for the higher order KdV are also constructed using the non-strict hyperbolicity 
of the equations.
In the case of strictly hyperbolic Whitham equations for the KdV, 
the oscillations (dispersive shock)
have uniform structure. However, in the case of non-strict hyperbolic Whitham equations
for the higher order KdV, 
an additional structure has been found in the dispersive shocks.
This new structure is similar to the one found here in the 
dispersive shocks of the mKdV equation.

The organization of the paper is as follows. In Section 2, we will study the
eigenspeeds, $\mu_{g,1},\mu_{g,2},\mu_{g,3}$ and $\mu_{g,4}$ of the Whitham equations (\ref{mkdvw}) for $g=1$.
In Section 3, we will construct the self-similar solutions of the single phase Whitham equations
for the initial function (\ref{step1}) with $a > b > c$.
In Section 4, we will construct the self-similar solution of the Whitham equations
for the initial function (\ref{step1}) with $a > c > b$. 
In Section 5, we will briefly discuss how to handle the other step-like
initial data (\ref{step2}).


\section{The Whitham Equations}\label{whithamsec}

In this section we define the eigenspeeds $\lambda_{g,i}$'s and $\mu_{g,i}$'s  of the Whitham equations (\ref{nlsw}) and 
(\ref{mkdvw}) with $g=1$ for
the NLS and the mKdV equations. For simplicity, we suppress the subscript $g=1$ in the notation $\lambda_{g,i}$'s and $\mu_{g,i}$'s in the
rest of the paper. 
 
We first introduce the 
polynomials of $\xi$ for $n=0, 1,
2, \dots$ \cite{dub, for, pav}:
\begin{equation} \label{eq15} 
P_n(\xi, u_1, u_2, u_3, u_4) = \xi^{n+2} + a_{n, 1} \xi^n + \dots + a_{n,
  n+2} \ ,
\end{equation}
where the coefficients, $a_{n, 1}, a_{n, 2}, \dots, a_{n, n+2}$ are
uniquely determined by the two conditions
\begin{equation*} 
\frac{ P_n(\xi, u_1, u_2, u_3, u_4) }{\sqrt{ (\xi - u_1)(\xi-u_2)(\xi-u_3)(\xi - u_4)}
} = \xi^{n} + \mathcal{O}(\xi^{-2}) \quad  \mbox{for large $|\xi|$} 
\end{equation*}
and
\begin{equation*} 
\int_{u_2}^{u_1} \frac{ P_n(\xi, u_1, u_2, u_3, u_4)}{\sqrt{ (u_1-\xi)
    (\xi-u_2)(\xi-u_3)(\xi - u_4)}} d\xi = 0 \ .
\end{equation*}
The coefficients of $P_n$ can be expressed in terms of
complete elliptic integrals.

The eigenspeeds of the Whitham equations (\ref{nlsw}) with $g=1$ for the NLS equation are 
defined in terms of
$P_0$ and $P_1$ of (\ref{eq15}) \cite{for, jin, pav}, 
\begin{equation*}  
\lambda_i(u_1, u_2, u_3, u_4) = 8 \ \frac{
  P_1(u_i, u_1, u_2, u_3, u_4) }{P_0(u_i, u_1, u_2, u_3, u_4)} \ , \quad i=1,2,3,4 \ ,
\end{equation*}
which give 
\begin{equation}
\label{lambda}
\lambda_i(u_1, u_2, u_3, u_4) = 2 \left(\sigma_1(u_1,u_2,u_3,u_4) - {I(u_1, u_2, u_3, u_4) \over \partial_{u_i} 
I(u_1, u_2, u_3, u_4)} \right) \,.
\end{equation}
Here $\sigma_1:=\sum_{j=1}^4u_j$, and $I(u_1, u_2, u_3, u_4)$ is given by a complete elliptic integral \cite{Tian1}
\begin{equation}
\label{I}
I(u_1, u_2, u_3, u_4) = \int_{u_2}^{u_1} {d \eta \over \sqrt{(u_1 - \eta)(\eta - u_2)(\eta - u_3) (\eta - u_4)}}
\ .
\end{equation}
The function $I$ can be rewritten as a contour integral. Hence,
\begin{equation}
\label{EPD}
2(u_i - u_j) {\pd^2 I \over \pd u_i \pd u_j} = {\pd I \over \pd u_i} - {\pd I \over \pd u_j} \ ,
\quad i, j = 1, 2, 3, 4 \ 
\end{equation}
since the integrand satisfies the same equations for each $\eta \neq u_i$.
This contour integral connection also allows us to give another formulation of $I$
\begin{equation}
\label{I2}
I(u_1, u_2, u_3, u_4) = \int_{u_4}^{u_3} {d \eta \over \sqrt{(u_1 - \eta)(u_2 - \eta)(u_3 - \eta) (\eta - u_4)}}
\ .
\end{equation}
It follows from (\ref{lambda}), (\ref{I}) and (\ref{I2}) that 
\begin{equation}
\label{strict}
\lambda_4 - 2 \sigma_1 < \lambda_3 - 2 \sigma_1 < 0 
 < \lambda_2 - 2 \sigma_1 < \lambda_1 - 2 \sigma_1 
\end{equation}
for $u_4 < u_3 < u_2 < u_1$. This implies the strict hyperbolicity of the NLS-Whitham equation (\ref{nlsw}) for $g=1$.

The eigenspeeds $\lambda_i$'s have the following values \cite{Tian1}:
At $u_3=u_4$, we have
\begin{eqnarray}
\label{bc3}
\left\{\begin{array}{lll}
\lambda_1= 6 u_1 + 2 u_2, \\
\lambda_2= 2u_1 + 6u_2, \\
\lambda_3 = \lambda_4 =2(u_1 + u_2 + 2 u_4) - { 8(u_1-u_4)(u_2-u_4)
\over u_1 + u_2 - 2 u_4} ,  
\end{array}
\right.
\end{eqnarray}
and at $u_2=u_3$, 
\begin{eqnarray}
\label{bc4}
\left\{\begin{array}{lll}
\lambda_1=6 u_1+2 u_4 ,  \\
\lambda_2=\lambda_3=2 u_1+4 u_3+2 u_4,   \\
\lambda_4=2 u_1+6 u_4  . 
\end{array}
\right.
\end{eqnarray}
Notice that the eigenspeed $\lambda_2=\lambda_3$ at $u_2=u_3$ is the same as the
velocity of the periodic solution (\ref{periodicsolution}), i.e. $V_1=2\sigma_1=2(u_1+2u_3+u_4)$.

The eigenspeeds of the mKdV-Whitham equations (\ref{mkdvw}) with $g=1$
are \cite{jin}
\begin{equation} \label{eq18a}
\mu_i(u_1, u_2, u_3,u_4) = 3 \ \frac{                                                                                                         
  P_2(u_i, u_1, u_2, u_3, u_4) }{P_0(u_i, u_1, u_2, u_3, u_4)} \ , \quad i=1,2,3,4 \ .                                                                                       
\end{equation}
They can be expressed in
terms of $\lambda_1$, $\lambda_2$, $\lambda_3$ and $\lambda_4$ of
the NLS-Whitham equations (\ref{nlsw}) with $g=1$.


\begin{lemma} \label{lem2.1}
The eigenspeeds $\mu_i(u_1,u_2,u_3,u_4)$'s of (\ref{eq18a}) can be expressed in the form
\begin{equation} \label{eq19}
\mu_i = 
\frac{1}{2} \left( \lambda_i - 2\sigma_1 \right)
\frac{\pd q}{\pd u_i} + 
q \ , \quad i=1,2,3,4\,,
\end{equation}
where $\sigma_1=\sum_{j=1}^4u_j$ and $q=q(u_1, u_2, u_3, u_4)$ is the solution of the boundary value problem
of the Euler-Poisson-Darboux equations
\begin{align} \label{eq20}
2 (u_i - u_j) \frac{\pd^2 q}{\pd u_i \pd u_j} &=
\frac{\pd q}{\pd u_i} - \frac{\pd q}{\pd u_j} \ , \quad i,j=1,2,3,4  \ ,\\
q(u, u, u) &= 3 u^2 . \nonumber
\end{align}
Also the $\mu_i$'s satisfy the over-determined systems
\begin{equation}
\label{a23}
\frac{1}{\mu_i-\mu_j}\,{\pd \mu_i \over \pd u_j} = \frac{1}{\lambda_i-\lambda_j}\,{\pd \lambda_i \over \pd u_j} \,,       
\quad \quad i \neq j \ .
\end{equation}

\end{lemma}

We omit the proof since it is very similar to the proof of an analogous result for
the KdV hierarchy \cite{Tian3}.

The boundary value problem (\ref{eq20}) has a unique solution.
The solution is a symmetric quadratic function of $u_1$, $u_2$,
$u_3$ and $u_4$
\begin{equation}
\label{q} q = \frac{3}{8}\sigma_1^2-\frac{1}{2}\sigma_2 \ ,
\end{equation}
where $\sigma_2=\sum_{i>j}u_iu_j$ is the elementary symmetric polynomial of degree two.
Notice that $q$ gives the velocity of the periodic solution (\ref{periodicsolution2}) for the
mKdV equation, i.e. $V_2=q$.

For NLS, $\lambda_i$'s satisfy \cite{Tian1} 
\begin{equation}
\label{34}
\frac{\pd \lambda_4 }{\pd u_4} < \frac{3}{2}
\frac{\lambda_3 - \lambda_4 }{u_3 - u_4} < \frac{\pd
  \lambda_3}{\pd u_3}  
\end{equation}
for $u_4 < u_3 < u_2 < u_1 $. Similar results also hold for the mKdV-Whitham equations (\ref{mkdvw}) with $g=1$. 
\begin{lemma}
\begin{eqnarray} 
\frac{\pd \mu_3}{\pd u_3} &>& \frac{3}{2}  \ \frac{\mu_3 - \mu_4}{u_3 - u_4}  \quad \mbox{if} \  \ 
\frac{\pd q}{\pd u_3} > 0 \label{ine1} \ , \\
\frac{\pd \mu_4}{\pd u_4} &<& \frac{3}{2}  \ \frac{\mu_3 - \mu_4}{u_3 - u_4}  \quad \mbox{if} \ \ 
\frac{\pd q}{\pd u_4} > 0 \label{ine2} \ ,
\end{eqnarray}
for $u_4 < u_3 < u_2 < u_1$.
\end{lemma}
\begin{proof}
We use (\ref{eq19}) and (\ref{34}) to obtain 
\begin{align} \nonumber
\frac{\pd \mu_3}{\pd u_3} &= \frac{1}{2}
\frac{\pd \lambda_3}{\pd u_3} 
\frac{\pd q}{\pd u_3} + \frac{1}{2} (\lambda_3
-2\sigma_1)\frac{\pd^2 q}{\pd u_3^2} \\
& > \frac{3}{4} \frac{\lambda_3 - \lambda_4}{u_3 - u_4}
\frac{\pd q}{\pd u_3} + \frac{1}{2} (\lambda_4 -
2\sigma_1) \frac{\pd^2 q}{\pd u_3^2} \ , \label{eq20a} 
\end{align}
and 
\begin{align}
\mu_3 - \mu_4 &= \frac{1}{2} \left(
\lambda_3 - \lambda_4 \right) \frac{\pd q}{\pd
  u_3} + \frac{1}{2} ( \lambda_4 - 2\sigma_1
) \left( \frac{\pd q}{\pd u_3} - \frac{\pd
  q}{\pd u_4}\right) \nonumber \\
&= \frac{1}{2} \left( \lambda_3 - \lambda_4 \right)
\frac{\pd q}{\pd u_3} + (\lambda_4 -
2\sigma_1) (u_3 - u_4) \frac{\pd^2 q}{\pd
  u_3 \pd u_4} \nonumber \\
&= \frac{2}{3} (u_3 - u_4) \left( \frac{3}{4}\, \frac{\lambda_3 -
  \lambda_4 }{u_2 - u_3} \,\frac{\pd q}{\pd u_3} 
+ \frac{3}{2} (\lambda_4 - 2\sigma_1)
\frac{\pd^2 q}{\pd u_3 \pd u_4} \right), \label{eq20b}
\end{align}
where we have used equation (\ref{eq20})   
\begin{displaymath} \label{eq21}  
\frac{\pd q}{\pd u_3} - \frac{\pd q}{\pd u_4} 
= 2(u_3 - u_4) \frac{\pd^2 q}{\pd u_3 \pd u_4} . 
\end{displaymath}
It follows from formula (\ref{q}) for $q$ that
\begin{equation*}
3 \frac{\pd^2 q}{\pd u_3 \pd u_4 } = \frac{\pd^2 q}{\pd u_3^2} \ , 
\end{equation*}
which, along with with (\ref{eq20a}) and (\ref{eq20b}), proves (\ref{ine1}).
Inequality (\ref{ine2}) is proved in the same way.
\end{proof}

The following calculations are useful in the subsequent sections.
Using formula (\ref{eq19}) for $\mu_3$ and $\mu_4$ and formulae (\ref{lambda}) for
$\lambda_3$ and $\lambda_4$, we obtain
\begin{eqnarray}
\mu_3 - \mu_4 &=& {I \over (\pd_{u_3} I) 
(\pd_{u_4} I)} \left[ {\pd q \over \pd u_4} {\pd I \over \pd u_3} - 
{\pd q \over \pd u_3} {\pd I \over \pd u_4} \right]  \no \\
&=& {I \over (\pd_{u_3} I) (\pd_{u_4} I)} \left[{\pd q \over \pd u_4} 
({\pd I \over \pd u_3} - {\pd I \over \pd u_4}) - ({\pd q \over \pd u_3} -
{\pd q \over \pd u_4}) {\pd I \over \pd u_4} \right] \no \\
&=& {2 I (u_3 - u_4) \over (\pd_{u_3} I) (\pd_{u_4} I)} \ M \ ,
\label{M}
\end{eqnarray}
where    
$$M = {\pd q \over \pd u_4} \ {\pd^2 I \over \pd u_3 \pd u_4}
- {\pd^2 q \over \pd u_3 \pd u_4} \ {\pd I \over \pd u_4} \ .$$
Here we have used equations (\ref{EPD}) for $I$ and equations (\ref{eq20}) for $q$
in equality (\ref{M}).
Since $q$ of (\ref{q}) is quadratic, we obtain 
\begin{equation}
\label{pM}
{\pd M \over \pd u_3} = {\pd q \over \pd u_4} 
 \ {\pd^3 I \over \pd u_3^2 \pd u_4} \ .
\end{equation}
We note that another expression for $M$ is
$$M = {\pd q \over \pd u_3} \ {\pd^2 I \over \pd u_3 \pd u_4}          
- {\pd^2 q \over \pd u_3 \pd u_4} \ {\pd I \over \pd u_3} \ .$$ 
Hence, we get
\begin{equation}
\label{pM'}
{\pd M \over \pd u_4} = {\pd q \over \pd u_3}                          
\ {\pd^3 I \over \pd u_3 \pd u_4^2} \ .                                                   
\end{equation}

We next evaluate $M(u_1, u_2, u_3, u_4)$ when $u_3 = u_4$.
Using the integral formula (\ref{I}) for the function I and applying the change of variable 
$\eta = (u_1 - u_2) \nu + u_2$, we obtain
\begin{align*}
M \Big|_{u_3=u_4} &= { {\pd q \over \pd u_4} \over 4 (u_2 - u_4)^3}  
\int_0^1 {d \nu \over (1 + {u_1 - u_2 \over u_2 - u_4} \nu )^3 
\sqrt{\nu (1 - \nu)} } \\ 
&\phantom{=} - 
{ {\pd^2 q \over \pd u_3 \pd u_4} \over 2 (u_2 - u_4)^2}
\int_0^1 {d \nu \over (1 + {u_1 - u_2 \over u_2 - u_4} \nu )^2                          
\sqrt{\nu (1 - \nu)} } \ .
\end{align*}
The two integrals can be evaluated exactly as
$$\int_0^1 { d \nu \over (1 + \gamma \nu )^3 \sqrt{\nu (1 - \nu)} } 
= { \pi (8 + 8 \gamma + 3 \gamma^2) \over 8 (1 + \gamma )^{5 \over 2}} \ , \quad
\int_0^1 { d \nu \over (1 + \gamma \nu )^2 \sqrt{\nu (1 - \nu)} }                     
= { \pi (2 +  \gamma) \over 2 (1 + \gamma )^{3 \over 2}} $$
for $\gamma > -1$. We finally get
\begin{equation}
\label{M44}
M \Big|_{u_3=u_4} = { \pi U(u_1, u_2, u_4) \over 128 
[(u_2 - u_4)(u_1 - u_4)]^{5 \over 2}} \ ,
\end{equation}
where 
\begin{eqnarray}
U(u_1, u_2, \xi) &=& 
[8(u_2 - \xi)^2 + 8 (u_2 - \xi)(u_1 - u_2) + 3 (u_1 - u_2)^2]
(u_1 + u_2 + 4 \xi) \no \\
&& \quad  - \ 8 (u_2 - \xi) [2(u_2 - \xi)^2 
+ 3 (u_2 - \xi)(u_1 - u_2)
\no \\ && \quad\quad\quad
+ (u_1 - u_2)^2] \label{U} \ .  
\end{eqnarray}
Similar to (\ref{M}) for $\mu_3$ and $\mu_4$, we have
\begin{equation}
\mu_2- \mu_3
= {2 I (u_2 - u_3) \over (\pd_{u_2} I) (\pd_{u_3} I)} \ N \ ,
\label{N}
\end{equation}
where
$$N = {\pd q \over \pd u_2} \ {\pd^2 I \over \pd u_2 \pd u_3}
- {\pd^2 q \over \pd u_2 \pd u_3} \ {\pd I \over \pd u_2} \ .$$
Since $q$ of (\ref{q}) is quadratic, we obtain
\begin{equation}
\label{pN}
{\pd N \over \pd u_3} = {\pd q \over \pd u_2}
 \ {\pd^3 I \over \pd u_2 \pd u_3^2} \ .
\end{equation}
  Finally, we use (\ref{bc3}) and (\ref{eq19}) to calculate 
\begin{eqnarray}
(\mu_2 - \mu_3)\Big|_{u_3=u_4} &=& 
{1 \over 2} [\lambda_2 - 2(u_1 + u_2 + 2u_4)] {\pd q \over \pd u_2} 
- {1 \over 2} [\lambda_3 - 2(u_1 + u_2 + 2u_4)] {\pd q \over \pd u_3} \no \\
&=& {(u_2 - u_4) \over 2(u_1 + u_2 - 2u_4)} \ V(u_1, u_2, u_4)  \label{V} \ ,
\end{eqnarray}
where 
\begin{equation}
V(u_1, u_2, u_4) = 3 u_1^2 + 3 u_2^2 - 12 u_4^2 + 6 u_1 u_2 + 6u_1u_4 - 6 u_2 u_4 \ .
\label{pV}
\end{equation}

\section{Self-Similar Solutions} 

In this section, we construct self-similar solutions of the Whitham equations 
(\ref{mkdvw}) with $g=1$ for the initial function (\ref{step1}) with $a > b > c$.
The case with $a > c > b$ will be studied in next section.
The solution of the zero phase Whitham equations (\ref{mkdv0}) 
does not develop a shock when 
$a + 5b \leq 0$.
We are 
therefore only interested in the case $a + 5b  > 0$. 

We first study the $\xi$-zero of the cubic polynomial equation
\begin{equation}
\label{cubic}
U(a, b, \xi) = 0 \ ,
\end{equation}
where $U$ is given by (\ref{U}). 
It is easy to prove that for each pair of $a$ and $b$ satisfying $a > b$ 
and $a + 5b > 0$, $U(a, b, \xi)=0$ has only one simple real root. Denoting this
zero by $\xi(a, b)$, we then deduce that 
$U(a, b, \xi)$ is positive for $\xi > \xi(a,b)$ and negative for $\xi < \xi(a,b)$. Since 
$U(a, b, -(a+b)/4) < 0$
in view of (\ref{U}), we must have 
\begin{equation}
\label{xx}
\xi(a,b) > -{a + b \over 4} \ .
\end{equation}
For initial function (\ref{step1}) with $a > b > c$ and $a + 5b > 0$, we now classify the 
resulting Whitham
solutions into four types:
\begin{itemize}
\item[(I)] $ ~ \xi(a, b) \leq c $ with any $a>b>c$
\item[(II)] $ ~ \xi(a, b) > c$ with $a+5b > 3(b-c)>0$
\item[(III)] $~ \xi(a, b) > c$ with $a+5b = 3(b-c)>0$
\item[(IV)] $ ~ \xi(a, b) > c$ with $0<a+5b < 3(b-c)$
\end{itemize}
We will study the second type (II) first. 

\subsection{Type II}
Here we consider the 
step initial function (\ref{step1}) 
satisfying $\xi(a, b) > c$ and $a+5b >
3(b-c) > 0$.  

\begin{theorem}(see Figure \ref{fig1}.)
\label{thm1}
For the step-like initial data (\ref{step1}) with $a > b > c$, $a + 5b>3(b-c)$ and $\xi(a, b) > c$,
the solution $(\alpha, \beta)$ of 
the zero phase Whitham equations (\ref{mkdv0}) and the solution $(u_1, u_2, u_3, u_4)$ of
the single phase Whitham equations
(\ref{mkdvw}) with $g=1$ are given as follows:
\begin{itemize}
\item[(1)] For $ x/t \leq \mu_3(a, b, \xi(a,b), \xi(a,b))$, 
\begin{equation}\label{bs1}
\alpha=a\,,\quad \beta=b\,.
\end{equation}
\item[(2)] For $\mu_3(a, b, \xi(a,b), \xi(a,b)) < x/t < \mu_3(a, b, u^{**}, c)$,
\begin{equation}
\label{ws1}
u_1 = a \ , \quad u_2 = b \ , \quad \frac{x}{t} = \mu_3(a, b, u_3, u_4) \ , 
\quad \frac{x}{t} = \mu_4(a, b, u_3,u_4) \ , 
\end{equation}
where $u^{**}$ is the unique   
solution $u_3$ of 
$\mu_3(a, b, u_3, c)=\mu_4(a, b, u_3, c)
$
in the interval 
$c <u_3< b$. 
\item[(3)] For $\mu_3(a, b, u^{**}, c)  \leq x/t  < \mu_3(a, b, b, c)$,
\begin{equation}
\label{ws2}
u_1 = a \ , \quad u_2 = b \ , \quad \frac{x}{t} = \mu_3(a, b, u_3, c) \ , \quad 
u_4 = c  .
\end{equation}
\item[(4)] For $ x/t\geq \mu_3(a, b, b, c)$,
\begin{equation}\label{bs2}
\alpha=a\,,\quad \beta=c.
\end{equation}
\end{itemize}
\end{theorem}

The boundaries $x/t = \mu_3(a, b, \xi(a,b), \xi(a,b))$ and $x/t = \mu_3(a, b, b, c)$ 
are called the trailing and leading edges, respectively.
They separate the solutions of the single phase Whitham equations (\ref{mkdvw}) with $g=1$ and 
the zero phase Whitham equations (\ref{mkdv0}).
The single phase Whitham solution matches the zero phase Whitham solution in the following 
fashion (see Figure 1.):
\begin{eqnarray}
\label{tr1}
(u_1, u_2) &=& \mbox{the solution $(\alpha, \beta)$ of (\ref{mkdv0}) defined outside the region} \ , \\
\label{tr2}
u_3 &=& u_4 \ , 
\end{eqnarray}
at the trailing edge; 
\begin{eqnarray}
\label{le1}
(u_1, u_4)  &=& \mbox{the solution $(\alpha, \beta)$ of (\ref{mkdv0}) defined outside the region} \ , \\
\label{le2}
u_2 &=& u_3 \ , 
\end{eqnarray}
at the leading edge. 

The proof of Theorem 3.1 is based on a series of lemmas:
We first show that the solutions defined by formulae (\ref{ws1}) and (\ref{ws2}) 
indeed satisfy the Whitham equations (\ref{mkdvw}) for $g=1$ \cite{dub, tsa}. 

\begin{lemma}

\begin{itemize}
\item[(1)] The functions $u_1$, $u_2$, $u_3$ and $u_4$ determined by equations (\ref{ws1})
give a solution of the Whitham equations (\ref{mkdvw}) with $g=1$ as long as $u_3$ and $u_4$
can be solved from (\ref{ws1}) as functions of $x$ and $t$.

\item[(2)] The functions $u_1$, $u_2$, $u_3$ and $u_4$ determined by equations (\ref{ws2})
give a solution of the Whitham equations (\ref{mkdvw}) with $g=1$ as long as $u_3$ 
can be solved from (\ref{ws2}) as a function of $x$ and $t$.
\end{itemize}

\end{lemma}

\begin{proof}
 (1) $u_1$ and $u_2$ obviously satisfy the first two equations of 
(\ref{mkdvw}) for $g=1$. To verify the 
third and fourth equations, we observe that 
\begin{equation}
\label{dia}
\frac{\pd \mu_3 }{\pd u_4} = \frac{\pd \mu_4 }{\pd u_3} = 0
\end{equation}
on the solution of (\ref{ws1}). To see this, we use (\ref{a23}) to calculate
$$\frac{\pd \mu_3 }{\pd u_4} = {{\pd \lambda_3 \over \pd u_4} \over \lambda_3 - 
\lambda_4}
\ (\mu_3 - \mu_4) = 0 \ .$$
The second part of (\ref{dia}) can be shown in the same way.
We then calculate the partial derivatives of the third equation of (\ref{ws1})
with respect to $x$ and $t$
$$ 1 = \frac{\pd \mu_3 }{\pd u_3} \ t u_{3x} \ , \quad 0 = \frac{\pd \mu_3 }{\pd u_3} \ t u_{3t} + \mu_3 \ ,$$ 
which give the third equation of (\ref{mkdvw}) with $g=1$. 
The fourth equation of (\ref{mkdvw}) with $g=1$ can be verified in the same way.

(2) The second part of Lemma 3.2 can easily be proved.
\end{proof}

We now determine the trailing edge. Eliminating $x$ and $t$ from the last two equations of (\ref{ws1})
yields 
\begin{equation}
\label{m34}
\mu_3(a, b, u_3, u_4) - \mu_4(a, b, u_3, u_4) = 0 \ . 
\end{equation}
Since it degenerates at $u_3 = u_4$, we replace (\ref{m34}) by 
\begin{equation}
\label{F1}
F(a, b, u_3,u_4) := {\mu_3(a, b, u_3,u_4)-\mu_4(a, b, u_3,u_4) \over u_3 - u_4} = 0 \ .
\end{equation}
Therefore, at the trailing edge where $u_3=u_4$, equation
(\ref{F1}), in view of formulae (\ref{M}) and (\ref{M44}), reduces to
\begin{equation}
\label{U2}
U(a, b, u_4) = 0 \ .
\end{equation}
Noting that $\xi(a,b)$ is the unique solution of (\ref{cubic}), we then deduce 
that $u_4 = \xi(a,b)$.

\begin{lemma}
Equation (\ref{F1}) has a unique solution satisfying $u_3=u_4$. The solution
is $u_3=u_4= \xi(a,b)$. The rest of equations (\ref{ws1}) at the trailing edge
are $u_1=a$, $u_2 = b$ and
$x/t = \mu_3(a, b, \xi(a,b), \xi(a,b))$.
\end{lemma}  

Having located the trailing edge, we now solve equations (\ref{ws1}) in the
neighborhood of the trailing edge. We first consider equation (\ref{F1}).
We use (\ref{M}) to write $F$ of (\ref{F1}) as 
$$F(a, b, u_3, u_4) = { 2 I \over  (\pd_{u_3} I) (\pd_{u_3} I) } M(a, b, u_3, u_4) \ . $$
We note that at the trailing edge $u_3=u_4= \xi(a,b)$, we have $M(a, b, \xi(a,b), \xi(a,b))= 0$
because of (\ref{M44}) and (\ref{U2}). We then use (\ref{pM}) and (\ref{pM'}) to
differentiate $F$ at the trailing edge 
\begin{align*}
{\pd F(a, b, \xi(a,b), \xi(a,b)) \over \pd u_3} &= 
{\pd F(a, b, \xi(a,b), \xi(a,b)) \over \pd u_4} \\ &= 
{ I \over 2 (\pd_{u_3} I) (\pd_{u_3} I) } \ [a + b + 4 \xi(a,b)] \
{\pd^3 I \over \pd u_3^2 \pd u_4} > 0 \ ,\end{align*}
where we have used the expression (\ref{q}) for $q$ in the last equation and (\ref{xx}) 
in the inequality.
These show that equation (\ref{F1}) or equivalently (\ref{m34}) can be 
inverted to give $u_4$ as a decreasing
function of $u_3$
\begin{equation}
\label{a} u_4 = A(u_3)
\end{equation}
in a neighborhood of $u_3=u_4= \xi(a,b)$. 

We now extend the solution $A(u_3)$ of equation (\ref{m34}) in the region
$c < u_4 < \xi(a,b) < u_3 < b$ as far as possible. We first claim that
\begin{equation}
\label{p>0} 
{\pd q(a, b, u_3, u_4) \over \pd u_3} > 0 \ , \quad
{\pd q(a, b, u_3, u_4) \over \pd u_4} > 0 
\end{equation}
on the extension. To see this, we first observe that inequalities (\ref{p>0}) are
true at the trailing edge $u_3=u_4=\xi(a,b)$. This follows from (\ref{q}) and (\ref{xx}).
Therefore, inequalities (\ref{p>0}) hold in a neighborhood 
of the trailing edge. To prove that (\ref{p>0}) remains true on the extension,
we use formula (\ref{eq19}) for $\mu_3$ and $\mu_4$
to rewrite equation (\ref{m34}) as
$$ {1 \over 2} [ \lambda_3 - 2(a + b + u_3 + u_4) ] {\pd q \over \pd u_3}
= {1 \over 2} [ \lambda_4 - 2(a + b + u_3 + u_4) ] {\pd q \over \pd u_4} \ .$$
Since the two terms in the two parentheses are both negative in view of (\ref{strict}) and since 
${\pd q \over \pd u_3} - 
{\pd q \over \pd u_4} = (u_3 - u_4)/4 >0$ in view of (\ref{q}), neither ${\pd q \over \pd u_3}$
nor ${\pd q \over \pd u_4}$ can vanish on the extension. This proves inequalities
(\ref{p>0}).  

We deduce from Lemma 2.2 that
\begin{equation}
\label{dia2}
{\pd \mu_3 \over \pd u_3} > 0 \ , \quad {\pd \mu_4 \over \pd u_4} <  0
\end{equation}
on the solution of (\ref{m34}).
Because of (\ref{dia}) and (\ref{dia2}), solution (\ref{a}) of equation (\ref{m34}) 
can be extended as long as $c < u_4 < \xi(a,b) < u_3 < b$.

There are two possibilities: (1) $u_3$ touches $b$ before or simultaneously
as $u_4$ reaches $c$ and (2) $u_4$ touches $c$ before $u_3$ reaches $b$.
It follows from (\ref{bc4}), (\ref{eq19}) and (\ref{q}) that
\begin{equation}
\label{u30} \mu_3(a, b, b, u_4) - \mu_4(a, b, b, u_4) 
= {1 \over 2} \ (b - u_4) (a + 2b + 3u_4)  > 0
 \quad \mbox{for $c \leq u_4 < b$} \ ,
\end{equation}
where we have used $a + 2b + 3c > 0$ in the inequality.
This shows that (1) is unattainable. Hence, $u_4$ will touch $c$ before $u_3$
reaches $b$. When this happens, equation (\ref{m34}) becomes 
\begin{equation}
\label{u**}
\mu_3(a,b,u_3,c)=\mu_4(a,b,u_3,c)\,.
\end{equation}

\begin{lemma}
\label{u**L}
Equation (\ref{u**}) has a simple zero in the interval $c  < u_3 < b$, counting
multiplicities. Denoting the zero by $u^{**}$, then $\mu_3(a, b, u_3, c)
 - \mu_4(a, b, u_3, c)$ is positive
for $u_3 > u^{**}$ and negative for $u_3 < u^{**}$.
\end{lemma}

\begin{proof}
We use (\ref{M}) and (\ref{pM}) to prove the lemma.
In both formulae, $\pd_{u_3} I$, $\pd_{u_4} I$ and $\pd^2_{u_3 u_4} I$ 
are all positive functions.
By (\ref{pM}), 
\begin{equation}
\label{PM1}
{\pd M(a, b, u_3, c) \over \pd u_3} = { (a + b + u_3 + 3c) \over 4} {\pd^3 I
\over \pd u_3^2 \pd u_4} \quad \ \mbox{for $c < u_3 < b$} \ .
\end{equation}
We claim that 
$$M(a, b, u_3, c) < 0 \quad \mbox{when $u_3 = c$ ~ ~ and ~ } 
~ M(a, b, u_3, c) > 0 \quad \mbox{for $u_3$ near $b$} \ .$$
The second inequality  follows from (\ref{M}) and (\ref{u30}).
The first inequality can be deduced from formula (\ref{M44})
$$M(a, b, c, c) = { \pi U(a, b, c) \over 128 [(b - c)(a - c)]^{5 \over 2}} 
<0 ~~ \quad \mbox{for $c < \xi(a, b)$} \ .$$
 Therefore, $M(a, b, u_3, c)$ has a zero in the interval $c < u_3 < b$. The uniqueness
of the zero follows from
(\ref{PM1}) in that $M(a, b, u_3, c)$ increases or changes from
decreasing to increasing as $u_3$ increases.
This zero is exactly
$u^{**}$ and the rest of the theorem can be proved easily.
\end{proof}

Having solved equation (\ref{m34}) for $u_4$ as a decreasing function of $u_3$
for $c < u_4 < \xi(a,b) < u_3 < b$, we turn to equations (\ref{ws1}). Because of (\ref{dia}) and (\ref{dia2}),
the third equation of (\ref{ws1}) gives $u_3$ as an increasing function of $x/t$, for
$\mu_3(a, b, \xi(a,b), \xi(a,b)) < x/t < \mu_3(a, b, u^{**}, c)$. 
Consequently, $u_4$ is a decreasing function of $x/t$ in the same interval.

\begin{lemma}
The last two equations of (\ref{ws1}) can be inverted to give $u_3$ and $u_4$ as 
increasing and decreasing functions, respectively, of the self-similarity variable
$x/t$ in the interval $\mu_3(a, b, \xi(a,b), \xi(a,b)) < x/t < \mu_3(a, b, u^{**}, c)$,
where $u^{**}$ is given in Lemma \ref{u**L}.
\end{lemma}

We now turn to equations (\ref{ws2}). We want to solve the third equation
when $x/t > \mu_3(a,b,u^{**}, c)$ or equivalently when $u_3 > u^{**}$. 
According to Lemma \ref{u**L},
$\mu_3(a, b, u_3, c) - \mu_4(a, b, u_3, c) > 0$ for $u^{**} < u_3 < b$.
In view of (\ref{p>0}), 
${\pd_{u_3} q}(a, b, u_3, c) = (a + b + 3u_3 + c)/4$ is positive at $u_3 =u^{**}$
and hence, it remains positive for $u_3 > u^{**}$. 
By (\ref{ine1}), we have
$${\pd \mu_3(a, b, u_3, c) \over \pd u_3} > 0 \ .$$
Hence, the third equation of (\ref{ws2}) can be solved for $u_3$ as an increasing
function of $x/t$ as long as $u^{**} < u_3 < b$. When $u_3$ reaches $b$, we have
$x/t = \mu_3(a, b, b, c)$.
We have therefore proved the following result.

\begin{lemma}
The third equation of (\ref{ws2}) can be inverted to give $u_3$ as an increasing
function of $x/t$ in the interval $\mu_3(a,b,u^{**}, c) \leq x/t \leq 
\mu_3(a, b, b, c)$.
\end{lemma}

We are ready to conclude the proof of Theorem 3.1.
 The solutions (\ref{bs1}) and (\ref{bs2}) are obvious.
According to Lemma 3.5, the last two equations of (\ref{ws1}) determine $u_3$ 
and $u_4$ as functions of 
$x/t$ in the region $\mu_3(a,b, \xi(a,b), \xi(a,b)) \leq x/t 
\leq \mu_3(a,b,u^{**}, c)$. By the first part of Lemma 3.2, the
resulting $u_1$, $u_2$, $u_3$ and $u_4$ satisfy the Whitham equations (\ref{mkdvw}) with $g=1$.
Furthermore, the boundary conditions (\ref{tr1}) and (\ref{tr2}) are satisfied
at the trailing edge $x = \mu_3(a, b, \xi(a,b), \xi(a,b))$.

Similarly, by Lemma 3.6, the third equation of (\ref{ws2}) determines $u_3$ 
as a function of $x/t$ in the region $\mu_3(a,b,u^{**}, c) \leq x/t 
\leq \mu_3(a, b, b, c)$. It then follows from
the second part of Lemma 3.2 that $u_1$, $u_2$, $u_3$ and $u_4$ of (\ref{ws2}) satisfy 
the Whitham equations (\ref{mkdvw}) for $g=1$.
They also satisfy the boundary conditions (\ref{le1}) and (\ref{le2}) at the
leading edge $x/t = \mu_3(a, b, b, c)$. 
We have therefore completed the proof of Theorem 3.1.  

A graph of the Whitham solution $(u_1,u_2,u_3,u_4)$ is given in Figure \ref{fig1}. It is obtained by
plotting  the exact solutions of (\ref{ws1}) and (\ref{ws2}).

\subsection{Type I}    
Here we consider
the initial function (\ref{step1}) satisfying $\xi(a,b) \leq
c$ with $b > c $ and $a + 5b > 0$.

\begin{figure}[h] 
\begin{center}
\includegraphics[width=12cm]{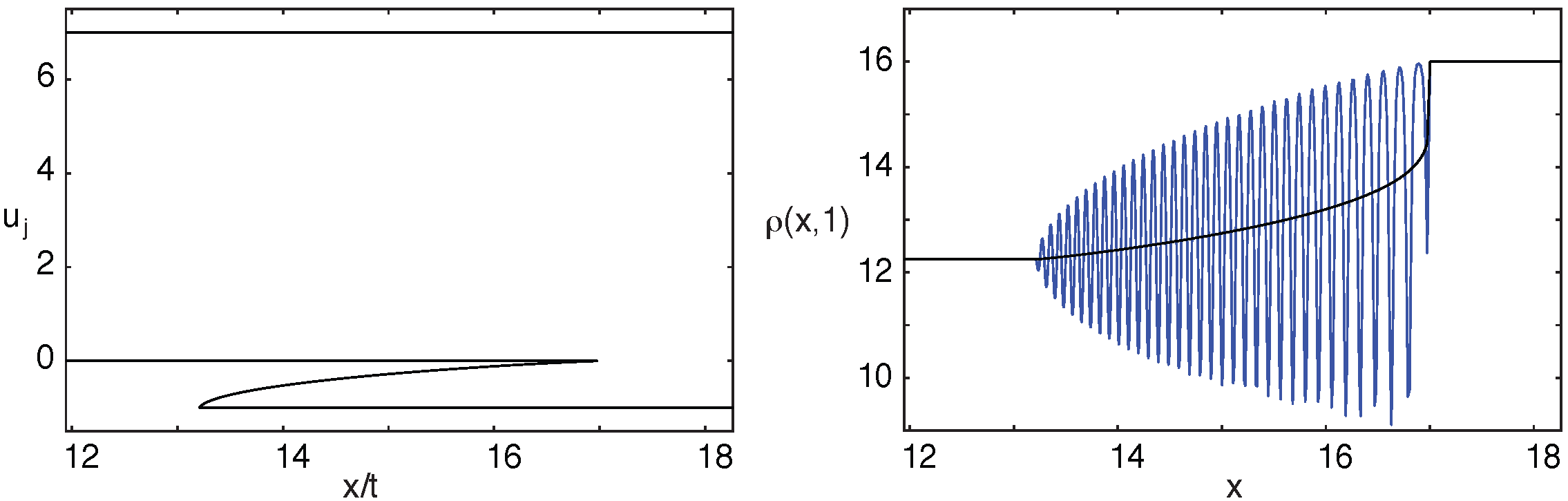}
\caption{\label{fig2} Self-Similar solution of the Whitham equations (\ref{mkdvw}) with $g=1$ and the
  corresponding oscillatory solution (\ref{periodicsolution2}) 
of the mKdV equation with
  $\epsilon = 0.07$.  The initial data is
  given by (\ref{step1}) with
  $a=7$, $b= 0$ and $c=-1$ of type I.}
\end{center}\end{figure}

\medskip

We will only present our proofs briefly, since they are, more or less, similar to those 
in Section 3.1. The main feature of this case is that the $\xi$-zero point does not appear
in the solution $u_3$, and the Whitham equations (\ref{mkdvw}) with $g=1$ are 
strictly hyperbolic on the solution.

\begin{theorem}(see Figure \ref{fig2}.)
For the step-like initial data (\ref{step1}) with $a > b > c$, $a + 5b > 0$ and $\xi(a, b) \leq c$,
the solution of the
Whitham equations (\ref{mkdvw}) with $g=1$ is given by 
\begin{equation*}
u_1 =a \ , \quad u_2 = b \ , \quad \frac{x}{t} = \mu_3(a, b, u_3, c) \ , \quad u_4 = c
\end{equation*}
for $\mu_3(a, b, c, c) < x/t < \mu_3(a,b, b, c)$.
Outside this interval, the solution of
(\ref{mkdv0}) is given by 
\begin{equation*}
\alpha = a \ , \quad \beta= b ~~ \quad \mbox{for}\quad\frac{x}{t} \leq \mu_3(a, b, c, c)\,,
\end{equation*}
and
\begin{equation*}
\alpha = a \ , \quad \beta = c ~~  \quad \mbox{for}\quad\frac{x}{t} \geq \mu_3(a,b,b, c)\ .
\end{equation*}
\end{theorem}

\begin{proof}
It suffices to show that $\mu_3(a, b, u_3, c)$ is an increasing function of $u_3$ for 
$c < u_3 < b$. 
Substituting (\ref{q}) for $q$ into (\ref{pM}) yields
$${\pd M(a, b, u_3, c) \over \pd u_3} = {1 \over 4} [a + b + u_3 + 3c] 
 \ {\pd^3 I \over \pd u_3^2 \pd u_4} \geq {1 \over 4} [a + b + 4 \xi(a,b)] 
\ {\pd^3 I \over \pd u_3^2 \pd u_4} > 0 $$
for $c < u_3 < b$, where we have used $\xi(a, b) \leq c $ in the first inequality
and (\ref{xx}) in the second one.
We now use formula (\ref{M44}) to calculate the value of $M(a,b, u_3, c)$ at $u_3=c$
$$M(a, b, c, c) = {\pi U(a, b, c) \over 128 [(b - c)(a - c)]^{5 \over 2}} 
 \geq 0 \quad \mbox{for $\xi(a,b) \leq c$} \ ,$$
because $U(a,b,\xi) \geq 0$ for $\xi \geq \xi(a,b)$. Therefore, 
$M(a, b, u_3, c) > 0$
for $c < u_3 < b$. It then follows from (\ref{M}) that 
$\mu_3(a, b, u_3, c) - \mu_4(a, b, u_3, c) > 0$.
Since ${\partial q \over \partial u_3}(a, b, u_3, c) = (a + b + 3u_3 + c)/4 >
(a + b + 3 \xi(a,b))/4 > 0$ because of (\ref{xx}), we conclude from Lemma 2.2 that 
$${d \mu_3(a,b, u_3, c) \over d u_3} > 0$$ for $c < u_3 < b$.
\end{proof}

\subsection{Type III} 
Here we consider
the step initial function (\ref{step1}) satisfying $\xi(a,b) > c$
with $a+5b = 3(b-c)> 0$.

\begin{figure}[h] 
\begin{center}
\includegraphics[width=12cm]{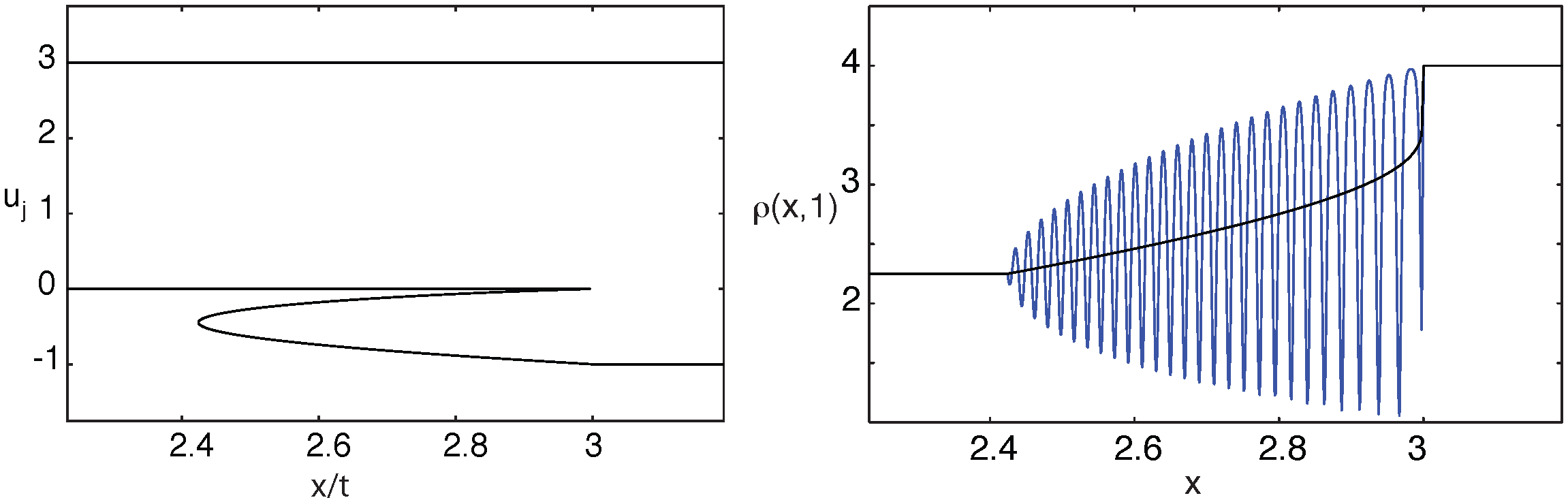}
\caption{\label{fig3} Self-Similar solution of the Whitham equations (\ref{mkdvw})
with $g=1$ and the
  corresponding oscillatory solution (\ref{periodicsolution2}) of 
the mKdV equation with
  $\epsilon = 0.007$.  The initial data is
  given by (\ref{step1}) with
  $a=3$, $b=0$, $c=-1$ of type III.}
\end{center}\end{figure}

\begin{theorem} (see Figure \ref{fig3}.)
For the step-like initial data (\ref{step1}) with $a > b > c $, $\xi(a,b) > c$ and
$a +5b=3(b-c)$,
the solution of the $g=1$ Whitham equations
(\ref{mkdvw}) with $g=1$ is given by
\begin{equation*}
u_1 = a \ , \quad u_2 = b \ , \quad \frac{x}{t} = \mu_3(a, b, u_3, u_4)  \ , 
\quad \frac{x}{t} = \mu_4(a, b, u_3, u_4)\,,
\end{equation*}
for $\mu_3(a, b, \xi(a,b), \xi(a,b)) < x/t < \mu_3(a, b, b, c)$.
Outside the region, the solution of equations (\ref{mkdv0}) is given
by
\begin{equation*}
\alpha = a \ ,\quad \beta = b\,, \quad \mbox{for}\quad\frac{x}{t} \leq 
\mu_3(a, b, \xi(a,b), \xi(a, b))\,,
\end{equation*}
and
\begin{equation*}
\alpha = a \ , \quad \beta = c\,, \quad \mbox{for}\quad\frac{x}{t} \geq \mu_3(a, b, b, c) \ .
\end{equation*}
\end{theorem}

\begin{proof}
It suffices to show that $u_3$ and $u_4$ of $\mu_2(a, b, u_3, u_4) - 
\mu_3(a, b, u_3, u_4)=0$
reaches $b$ and $c$, respectively, simultaneously. To see this, we deduce from
the calculation (\ref{u30}) that 
\begin{equation}
\label{M44''}
\mu_3(a, b, b, u_4) - \mu_4(a, b, b, u_4) 
= {1 \over 2} (b - u_4) (a + 2b + 3u_4) \ ,
\end{equation}
vanishes at $u_4 = (-a - 2b)/3 = c$.
\end{proof}

\subsection{Type IV}
Here we consider
the step initial function (\ref{step1}) satisfying $\xi(a,b) > c$ with
$0< a+5b < 3(b-c)$.

\begin{figure}[h] 
\begin{center}
\includegraphics[width=12cm]{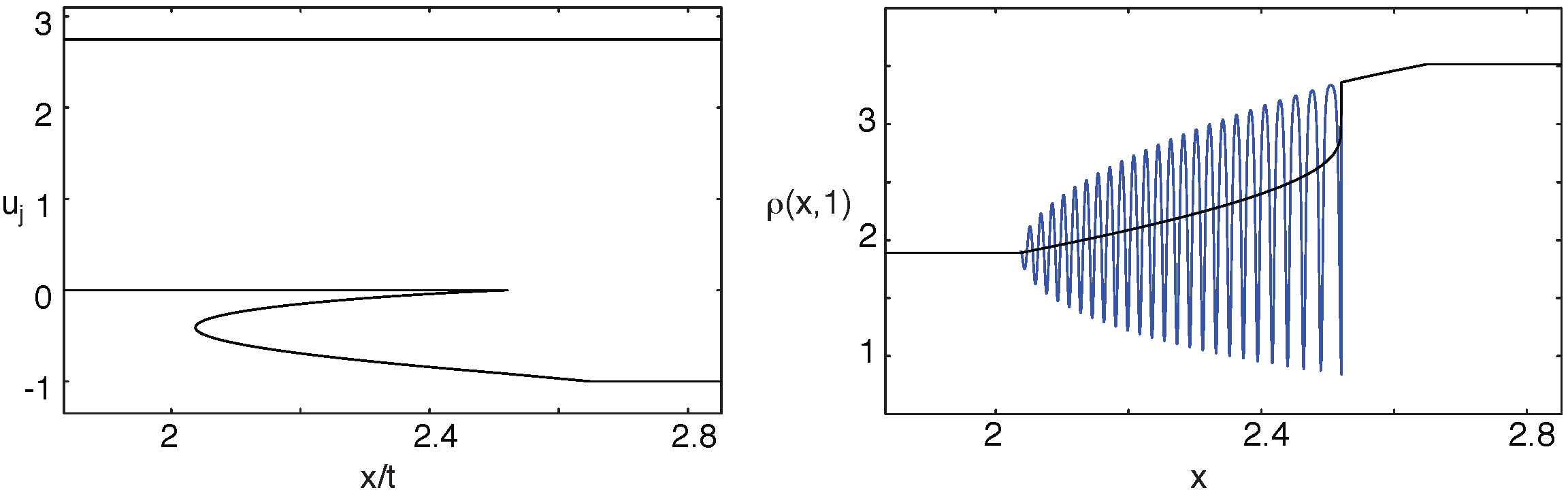}
\caption{\label{fig4} Self-Similar solution of the Whitham equations (\ref{mkdvw}) with $g=1$ and
the corresponding oscillatory solution (\ref{periodicsolution2}) of the mKdV equation with
$\epsilon = 0.006$. The initial data is
given by (\ref{step1}) with
  $a=2.75$, $b=0$ and $c=-1$ of type IV. The solution in the region $2.52<x/t<2.65$ 
  represents a rarefaction wave.}
\end{center}\end{figure}

\begin{theorem}(see Figure \ref{fig4}.)
For the step-like initial data (\ref{step1}) with $a > b > c$, $\xi(a, b) > c$ and 
$0<a+5b < 3(b-c)$,
the solution of the Whitham equations
(\ref{mkdv}) is given by
\begin{equation*}
u_1 = a \ , \quad u_2 = b \ , \quad \frac{x}{t} = \mu_2(a,b,u_3, u_4) \ , \quad 
x = \mu_3(a,b,u_3, u_4) \ t
\end{equation*}
for $\mu_3(a, b, \xi(a,b), \xi(a,b)) < x/t < \mu_3(a, b, b, -(a+2b)/3)$.
Outside the region, the solution of equation (\ref{mkdv0}) 
is divided into the three regions:
\begin{itemize}
\item[(1)] For $ x/t \leq 
\mu_3(a, b, \xi(a,b), \xi(a,b))$,
\begin{equation*}
\alpha =a \ , \quad \beta = b \ .
\end{equation*}
\item[(2)] For $\mu_3(a, b, b, - (a + 2b)/3) \leq \frac{x}{t} \leq \frac{3}{8} \left(
a^2 + 2 ac + 5c^2 \right)$,
\begin{equation*}
  \alpha=a, \quad \beta=
-\frac{1}{5} a - \sqrt{ \frac{8}{15} \frac{x}{t} - \frac{4}{25} a^2 } \ .
\end{equation*}
\item[(3)] For $ x/t \geq \frac{3}{8} \left( a^2 + 2ac + 5c^2 \right)$,
\[
\alpha=a, \quad \beta=c\ . \] 
\end{itemize}
\end{theorem}

\begin{proof}
By the calculation (\ref{M44''}), when $u_3$ of $\mu_3(a, b, u_3, u_4) - 
\mu_4(a, b, u_3, u_4)=0$
touches $b$, the corresponding $u_4$ reaches $-(a + 2b)/3$, which is above $c$. 
Hence, equations 
$$\frac{x}{t} = \mu_3(a, b, u_3, u_4)  \ , \quad \frac{x}{t} = \mu_4(a, b, u_3, u_4)$$
can be inverted to give $u_3$ and $u_4$ as functions of $x/t$ in the region 
$\mu_3(a, b, \xi(a,b), \xi(a,b)) < x/t < \mu_3(a, b, b, - (a + 2b)/3)$. 
In the region (2), the
equations (\ref{mkdv0}) has a rarefaction wave solution.
\end{proof}

\section{More Self-Similar Solutions}

In this section, we construct self-similar solutions of the $g=1$ Whitham equations
(\ref{mkdvw}) for the initial function (\ref{step1}) with $a > c > b$.
The solution of equations (\ref{mkdv0}) does not develop a shock for 
$a + 5b \geq 0$.
We are therefore only interested in the case $a + 5b < 0 $. We classify the resulting
Whitham solution into four types:
\begin{itemize}
\item[(V)]  $ ~ a + 5b\le -4(c-b)<0$
\item[(VI)] $ ~ 0>a + 5b > -4(c-b)$ with $V(a, c, b) < 0$
\item[(VII)] $ ~ 0>a +5b>-4(c-b)$ with $ V(a, c, b) = 0$
\item[(VIII)] $  ~0>a + 5b>-4(c-b)$ with $ V(a, c, b) > 0 $
\end{itemize}
where $V$ is a quadratic polynomial given by (\ref{pV}).

\subsection{Type V}
Here we consider
the step initial function (\ref{step1}) satisfying $a + 5b \leq
-4(c-b) < 0$.

\begin{figure}[h] 
\begin{center}
\includegraphics[width=12cm]{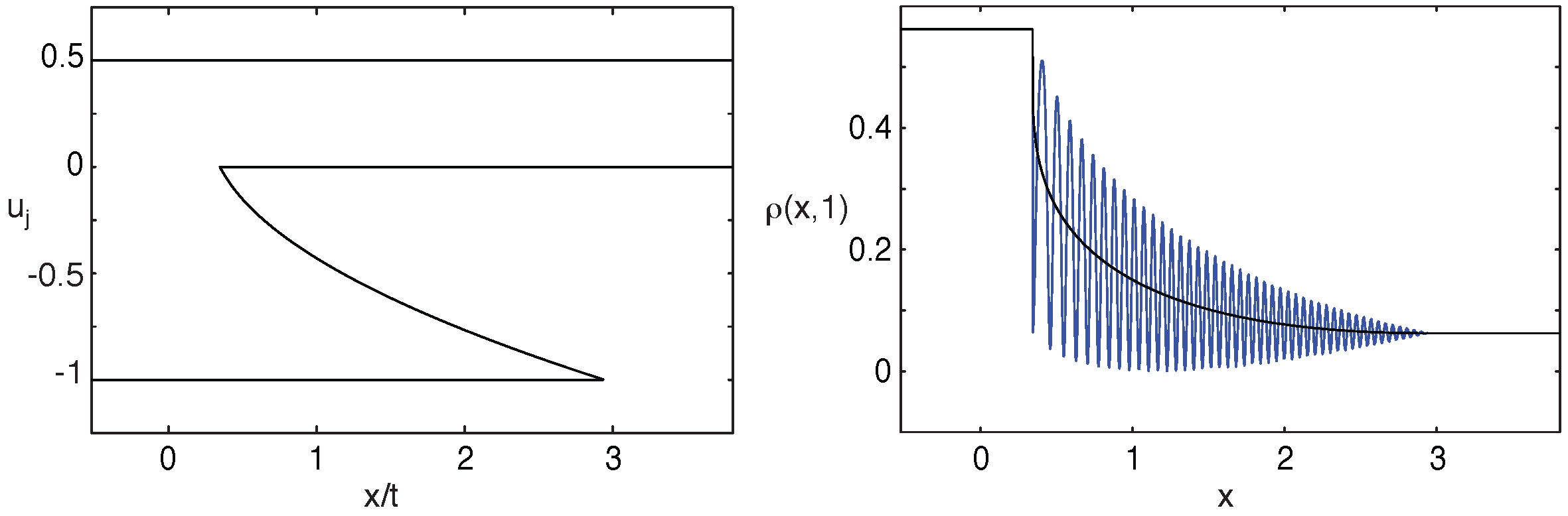}
\caption{\label{fig5} Self-Similar solution of the Whitham equations (\ref{mkdvw}) with $g=1$ and
the corresponding oscillatory solution (\ref{periodicsolution2}) of the mKdV equation with
$\epsilon = 0.018$.  The initial data is
given by (\ref{step1}) with
  $a= 1/2$, $b=-1$ and $c=0$ of type V.}
\end{center}\end{figure}

\begin{theorem}(see Figure \ref{fig5}.)
For the step-like initial data (\ref{step1}) with $a > c > b$,  $ a+5b\le -4(c-b)$, 
the solution of the
Whitham equations (\ref{mkdvw}) with $g=1$ is given by
\begin{equation*}
u_1 = a \ , \quad u_2 = c \ , \quad \frac{x}{t} = \mu_3(a, c, u_3, b) \ , \quad u_4 = b
\end{equation*}
for $\mu_3(a, c, c, b) < x/t < \mu_3(a, c, b, b)$. 
Outside this interval, the solution of
(\ref{mkdv0}) is given by
\begin{equation*}
\alpha = a \, \quad \beta= b  \quad \mbox{for}\quad\frac{x}{t} \leq \mu_3(a, c, c, b)\,,
\end{equation*}
and
\begin{equation*}
\alpha = a \, \quad \beta = c  \quad \mbox{for}\quad\frac{x}{t}\geq \mu_3(a, c, b, b) \ .
\end{equation*}
\end{theorem}

\begin{proof}
It suffices to show that $\mu_3(a, c, u_3, b)$ is a decreasing function of $u_3$ for
$b < u_3 < c$. By (\ref{eq19}),
we have
$${\pd \mu_3(a, c, u_3, b) \over \pd u_3} = {1 \over 2} {\pd \lambda_3 \over \pd u_3}
{\pd q \over \pd u_3} + {1 \over 2} [\lambda_3 - 2(a + c + u_3 + b)] {\pd^2 q \over 
\pd u_3^2} \ .$$
The second term is negative because of (\ref{strict}) and ${\pd^2 q \over \pd u_2^2} = 3/8 > 0$.
The first term is also negative: Its first factor is positive in view
of (\ref{eq11}); while its 
 second factor is
$${\pd q \over \pd u_3} = {1 \over 4}(a + c + 3u_3 + b) < 0,$$
for $b < u_3 < c$, as we have that $a + b + 4c \leq 0$.
\end{proof}

\subsection{Type VI}
Here we consider
the step initial function (\ref{step1}) satisfying $0 > a + 5b >
-4(c-b)$ with $V(a, c, b) < 0 $.

\begin{figure}[h] 
\begin{center}
\includegraphics[width=12cm]{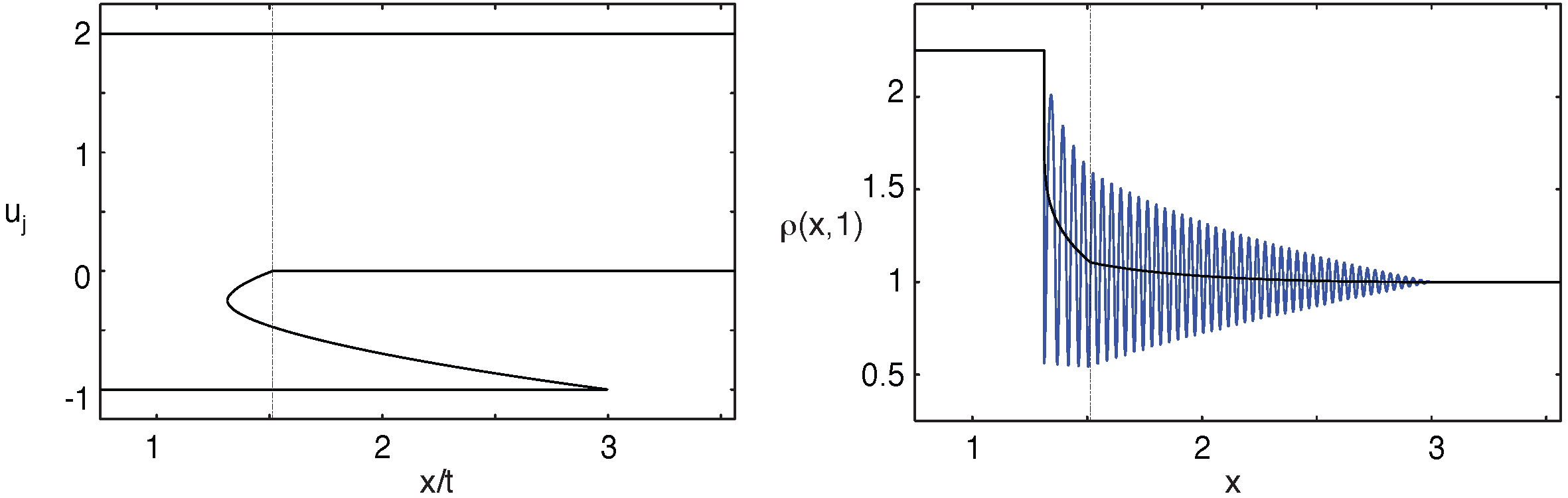}
\caption{\label{fig6} Self-Similar solution of the Whitham equations (\ref{mkdvw}) with $g=1$ and the
  corresponding periodic oscillatory solution (\ref{periodicsolution2}) of the mKdV equation with
  $\epsilon = 0.018$.  The initial data is
  given by (\ref{step1}) with
  $a= 2$, $b=-1$ and $c=0$ of type VI.  The oscillations have
  two distinct structures, which are separated by  $x/t \approx 1.51$.}
\end{center}\end{figure}

\begin{theorem}(see Figure \ref{fig6}.)
For the step-like initial data (\ref{step1}) with 
$0>a + 5b>-4(c-b)$ and $V(a, c, b)<0$,
the solution of the Whitham equations
(\ref{mkdvw}) with $g=1$ is given by
\begin{equation}
\label{ws1''}
u_1 = a \ , \quad \frac{x}{t} = \mu_2(a, u_2, u_3, b)  \ , \quad \frac{x}{t} = \mu_3(a, u_2, u_3, b)
\ , \quad u_4 = b
\end{equation}
for $\mu_3(a, - (a+b)/4, -(a+b)/4, b) < x/t \leq \mu_3(a, u^{***}, u^{***}, b)$ and by
\begin{equation}
\label{ws2''}
u_1 = a, \quad u_2 = c \ , \quad \frac{x}{t} = \mu_3(a, c, u_3, b)  \ , \quad u_4 = b
\end{equation}
for $\mu_3(a, u^{***}, u^{***}, b) \leq x/t < \mu_3(a, c, b, b)$, where 
$u^{***}$ is the unique solution
$u_3$ of $\mu_2(a, c, u_3, b) = \mu_3(a, c, u_3, b)$ in the interval 
$b < u_3 < c$.
Outside the region $ \mu_3(a, -(a+b)/4, -(a+b)/4, b) < x/t < \mu_3(a, c, b, b)$, 
the solution of 
equation (\ref{mkdv0}) is given by
\begin{equation*}
\alpha = a \ , \quad \beta = b \quad \mbox{for}\quad\frac{x}{t} \leq 
\mu_3(a, -(a+b)/4, -(a+b)/4, b)\,,
\end{equation*}
and
\begin{equation*}
\alpha = a \ , \quad \beta = c  \quad \mbox{for}\quad\frac{x}{t} \geq \mu_3(a, c, b, b) \ .
\end{equation*}
\end{theorem}

\begin{proof}
We first locate the ``leading'' edge, i.e., the solution of equation (\ref{ws1''}) at 
$u_2=u_3$.
Eliminating $x/t$ from the first two equations of (\ref{ws1''}) yields
\begin{equation}
\label{mu23}
\mu_2(a, u_2, u_3, b) - \mu_3(a, u_2, u_3, b) = 0 \ .
\end{equation}
Since it degenerates at $u_2=u_3$, we replace (\ref{mu23}) by
\begin{equation}
\label{G}
G(a, u_2, u_3, b) := {\mu_2(a, u_2, u_3, b) - \mu_3(a, u_2, u_3, b) \over 
(u_2 - u_3)\sqrt{(u_1 - u_3)(u_2 - u_4)}I(a, u_2, u_3, b)} = 0 \ .
\end{equation}
In the Appendix, we show that,
at the ``leading'' edge $u_2=u_3$, we have
$$G(a, u_3, u_3, b) = 2 ({\pd q \over \pd u_2} + {\pd q \over \pd u_3}) = 0 \ ,$$
in view of (\ref{G3}), which along with (\ref{q}) gives $u_2=u_3 = -(a+b)/4$.
Having located the ``leading'' edge, we solve equation (\ref{G}) near $u_2 = u_3 = 
-(a+b)/4$.
We use formula (\ref{G3'}) to obtain
$${\pd G(a, -(a+b)/4, -(a+b)/4, b) \over \pd u_2} = {\pd G(a, -(a+b)/4, -(a+b)/4, b) 
\over \pd u_3} = 2 \ . $$
These show that equation (\ref{G}) gives $u_2$ as a decreasing function of $u_3$
\begin{equation}
\label{B}
u_2 = B(u_3)
\end{equation}
in a neighborhood of $u_2 = u_3 = -(a+b)/4$.

We now extend the solution (\ref{B}) of equation (\ref{mu23}) as far as possible in the region
$b < u_3 < -(a+b)/4 < u_2 < c$. We use formula (\ref{eq19}) to obtain
\begin{eqnarray*}
{\pd \mu_2 \over \pd u_2} &=& {1 \over 2} {\pd \lambda_2 \over \pd u_2}{\pd q \over \pd u_2}
+ {1 \over 2} [\lambda_2 - 2(a + u_2 + u_3 + b)] {\pd^2 q \over \pd u_2^2} \ , \\
{\pd \mu_3 \over \pd u_3} &=& {1 \over 2} {\pd \lambda_3 \over \pd u_3}{\pd q \over \pd u_3}
+ {1 \over 2} [\lambda_3 - 2(a + u_2 + u_3 + b)] {\pd^2 q \over \pd u_3^2} \ .
\end{eqnarray*}
In view of (\ref{eq11}) and (\ref{strict}), we have
\begin{eqnarray*}
{\pd \mu_2 \over \pd u_2} &>& 0 \quad \mbox{if} \ \  {\pd q \over \pd u_2} > 0 \ , \\
{\pd \mu_3 \over \pd u_3} &<& 0 \quad \mbox{if} \ \ {\pd q \over \pd u_3} < 0 \ .
\end{eqnarray*}
We claim that
\begin{equation}
\label{cl}
{\pd q \over \pd u_2} > 0 \ , \quad {\pd q \over \pd u_3} < 0
\end{equation}
on the
solution of (\ref{mu23}) in the region $b < u_3 < -(a+b)/4 < u_2 < c$. To see this,
we use formula (\ref{eq19}) to rewrite equation (\ref{mu23}) as
$${1 \over 2} [\lambda_2 - 2(a + u_2 + u_3 + b)] {\pd q \over \pd u_2}
= {1 \over 2} [\lambda_3 - 2(a + u_2 + u_3 + b )] {\pd q \over \pd u_3} \ .$$
This, together with
$${\pd q \over \pd u_2} - {\pd q \over \pd u_3} = 2(u_2 - u_3) {\pd^2 q \over \pd u_2 \pd u_3}
= {1 \over 2}(u_2 - u_3) > 0 $$
for $u_2 > u_3$, and inequalities (\ref{strict}),
proves (\ref{cl}).

Hence, the solution (\ref{B}) can be extended as long as $b < u_3 < -(a+b)/4 < u_2 < c$.
There are two possibilities; (1) $u_2$ touches $c$ before
$u_3$ reaches $b$ and (2) $u_3$ touches $b$ before or simultaneously as 
$u_2$ reaches $c$.

Possibility (2) is unattainable. To see this, we
use (\ref{V}) to write
\begin{equation}
\label{2}
\mu_2(a, u_2, b, b) - \mu_3(a, u_2, b, b) = {(u_2 - b) \over 2(a + u_2 - 2b)} \ 
V(a, u_2, b) \ ,
\end{equation}
which is negative for $ b < u_2 \leq c$ since $V(a, u_2, b)$ of (\ref{pV}) is an 
increasing function of $u_2$ and since $V(a, c, b) < 0$.
Therefore, $u_2$ will touch $c$ before $u_3$ reaches $b$. When this happens, we have
\begin{equation}
\label{mu23e}
\mu_2(a, c, u_3, b) - \mu_3(a, c, u_3, b) =0 \ .
\end{equation}

\begin{lemma}
Equation (\ref{mu23e}) has a simple zero, counting multiplicities, in the interval
$b < u_3 < c$. Denoting this zero by $u^{***}$, then $\mu_2(a, c, u_3, b) - \mu_3(a, c, u_3, b)$
is positive for $u_3 > u^{***}$ and negative for $u_3 < u^{***}$.
\end{lemma}

The proof, which involves formulae (\ref{N}) and (\ref{pN}), is rather similar to the
proof of Lemma \ref{u**}. We will omit it.

We now continue to prove Theorem 4.2. Having solved equation (\ref{mu23}) for $u_2$ as a
decreasing function of $u_3$ for $u^{***} < u_3 < -(a+b)/4$, we can then use the middle 
two equations
of (\ref{ws1''}) to determine $u_2$ and $u_3$ as functions of $x/t$ in the interval
$\mu_2(a, -(a+b)/4, -(a+b)/4, b) < x/t < \mu_2(a, c, u^{***}, b)$.

We finally turn to equations (\ref{ws2''}). We want to solve the third equation
of (\ref{ws2''}), $x/t = \mu_3(a, c, u_3, b)$, for $u_3 < u^{***}$. It is enough to 
show that
$\mu_3(a, c, u_3, b)$ is a decreasing function of $u_3$ for $u_3 < u^{***}$.
According to Lemma 4.3, $\mu_2(a, c, u_3, b) - \mu_3(a, c, u_3, b) <  0$ for 
$u_3 < u^{***}$.
Using formula (\ref{eq19}) for $\mu_2$ and $\mu_3$, we have
$${1 \over 2} [\lambda_2 - 2(a + c + u_3 + b)] {\pd q \over \pd u_2}
< {1 \over 2} [\lambda_3 - 2(a + c + u_3 + b)] {\pd q \over \pd u_3} \ .$$
This, together with
$${\pd q \over \pd u_2} - {\pd q \over \pd u_3} = 
{1 \over 2}(c - u_3) > 0 $$
for $ u_3 < c$, and inequalities (\ref{strict}),
proves
$${\pd q(a, c, u_3, b) \over \pd u_3} < 0 $$
for $u_3 < u^{***}$.
Hence,
$${\pd \mu_3 \over \pd u_3} = {1 \over 2} {\pd \lambda_3 \over \pd u_3}{\pd q \over \pd u_3}
+ {1 \over 2} [\lambda_3 - 2(a + c + u_3 + b)] {\pd^2 q \over \pd u_3^2} < 0 \ ,$$
where we have used inequality (\ref{eq11}). 
\end{proof}

\subsection{Type VII}
Here we consider
the step initial function (\ref{step1}) satisfying $0> a+5b>
-4(c-b)$ with $V(a,c,b) = 0$.

\begin{figure}[h] 
\begin{center}
\includegraphics[width=12cm]{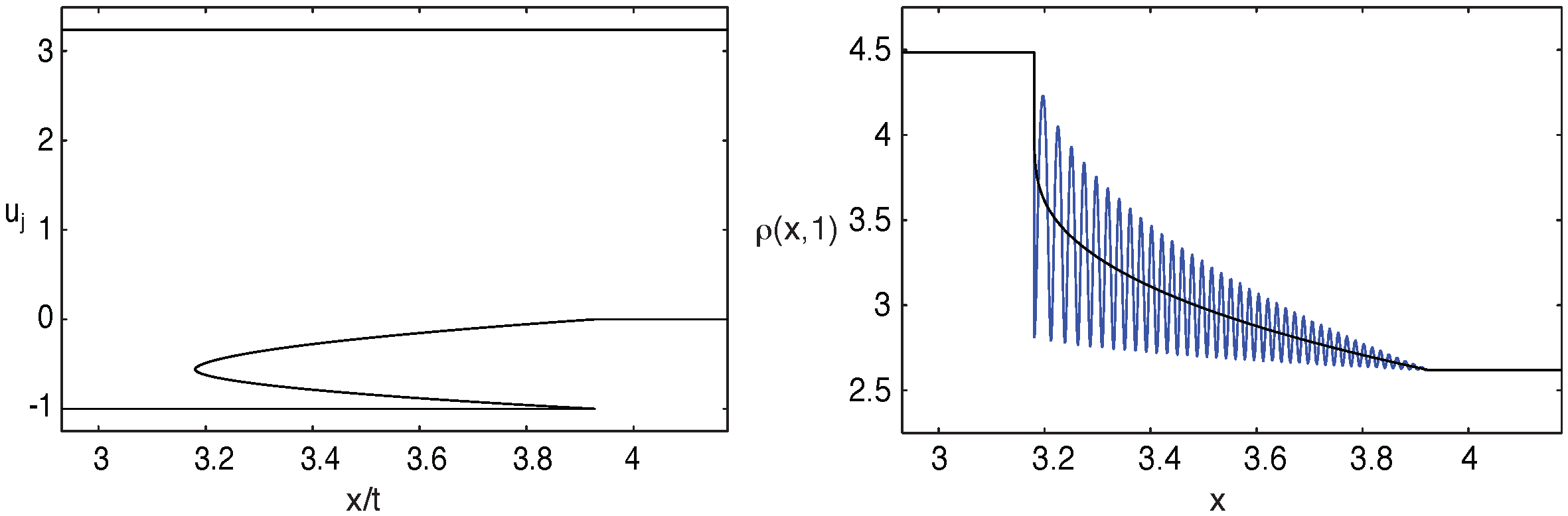}
\caption{\label{fig7} Self-Similar solution of the Whitham equations (\ref{mkdvw}) $g=1$ and
the corresponding oscillatory solution (\ref{periodicsolution2}) of the mKdV equation with
$\epsilon = 0.01$. The initial data is given
by (\ref{step1}) with
  $a=1+\sqrt{5}$, $b= -1$ and $c=0$ of type VII.}
\end{center}\end{figure}

\begin{theorem}(see Figure \ref{fig7}.)
For the step-like initial data (\ref{step1}) with $0>a+5b>-4(c-b)$ and $V(a,c,b)=0$,
the solution of the Whitham equations
(\ref{mkdvw}) with $g=1$ is given by
\begin{equation*}
u_1 = a \ , \quad\frac{x}{t} = \mu_2(a, u_2, u_3, b)  \ , \quad \frac{x}{t} = \mu_3(a, u_2, u_3, b) \ ,
\quad u_4 = c\,,
\end{equation*}
for $\mu_3(a, -(a+b)/4,-(a+b)/4, b)  < x/t < 
\mu_3(a, c, b,b)$.
Outside the region, the solution of 
equations (\ref{mkdv0}) is given by
\begin{equation*}
\alpha = a \ , \quad \beta = b\,,\quad \quad \mbox{for}\quad\frac{x}{t} \leq  
\mu_3(a, -(a+b)/4, -(a+b)/4, b)\,,
\end{equation*}
and
\begin{equation*}
\alpha = a \ , \quad \beta= c\,,\quad \quad ~~ \mbox{for}\quad\frac{x}{t} \geq \mu_3(a, c, b,b) \ .
\end{equation*}
\end{theorem}

\begin{proof}
It suffices to show that $u_2$ and $u_3$ of $\mu_2(a, u_2, u_3, b) - 
\mu_3(a, u_2, u_3, b)=0$
reaches $c$ and $b$, respectively, simultaneously. To see this, we deduce from
equation (\ref{2}) that
\begin{equation}
\label{7} 
\mu_2(a, c, b, b) - \mu_3(a, c, b, b) = { (c-b) \over 2 (a + c -2b)} \ V(a,c,b)
\end{equation}
vanishes when $ V(a, c, b) = 0$.
\end{proof}

\subsection{Type VIII}
Here we consider
the step initial function (\ref{step1}) satisfying
$0>a+5b>-4(c-b)$ with $V(a, c, b) > 0$.

\begin{theorem}(see Figure \ref{fig8}.)
For the step-like initial data (\ref{step1}) with $0>a+5b>-4(c-b)$ and $V(a, c, b) > 0$,
the solution of the Whitham equations
(\ref{mkdvw}) with $g=1$ is given by
\begin{equation*}
u_1 = a \ , \quad \frac{x}{t} = \mu_2(a, u_2, u_3, b)  \ , \quad \frac{x}{t} = \mu_3(a, u_2, u_3, b)  \ ,
\quad u_4 = b\,,
\end{equation*}
for $\mu_3(a, -(a+b)/4, -(a+b)/4, b)  < x/t < 
\mu_3(a, \hat{u}, \hat{u}, b)$, where $\hat{u}$ is the unique $u_2$-zero of 
the quadratic
polynomial $V(a, u_2, b)$ in the interval $-(a+b)/4 < u_2 < c$.
Outside the region, the solution
of equations (\ref{mkdv0}) is divided into the following three regions:
\begin{itemize}
\item[(1)] For $x/t\leq 
\mu_3(a, -(a+b)/4, -(a+b)/4, b)$,
\begin{equation*}
\alpha = a \ , \quad \beta = b\,.
\end{equation*}
\item[(2)] For $ \mu_2(a, \hat{u}, b, b) \leq x/t \leq \frac{3}{8} \left( a^2
+ 2 a c + 5 c^2 \right)$,
\begin{equation*}
\alpha=a\,,\quad \beta=
-\frac{1}{5} a + \sqrt{ \frac{8}{15} \frac{x}{t} - \frac{4}{25} a^2 }\,.
\end{equation*}
\item[(3)] For $x/t \geq \frac{3}{8} \left( a^2 + 2 a c + 5c^2 \right)$, 
\[
\alpha=a\,,\quad \beta=c\,.
\]
\end{itemize}
\end{theorem}

\begin{proof}
By the calculation (\ref{7}), when $u_3$ of $\mu_2(a, u_2, u_3, b) - 
\mu_3(a, u_2, u_3, b)= 0 $ 
touches $b$, the corresponding $u_2$ reaches $\hat{u}$, where $V(a, \hat{u}, b) =0$.
Obviously, $\hat{u} < c$.
Hence, equations
$$\frac{x}{t} = \mu_2(a, u_2, u_3, b) \,,\quad \quad \frac{x}{t} = \mu_3(a, u_2, u_3, b) $$
can be inverted to give $u_2$ and $u_3$ as functions of $x/t$ in the region
$\mu_2(a, -(a+b)/4, -(a+b)/4, b) < x/t < \mu_2(a, \hat{u}, \hat{u}, b)$.  
To the right of this region,
equations (\ref{mkdv0}) has a rarefaction wave solution.
\end{proof}

\begin{figure}[h] 
\begin{center}
\includegraphics[width=12cm]{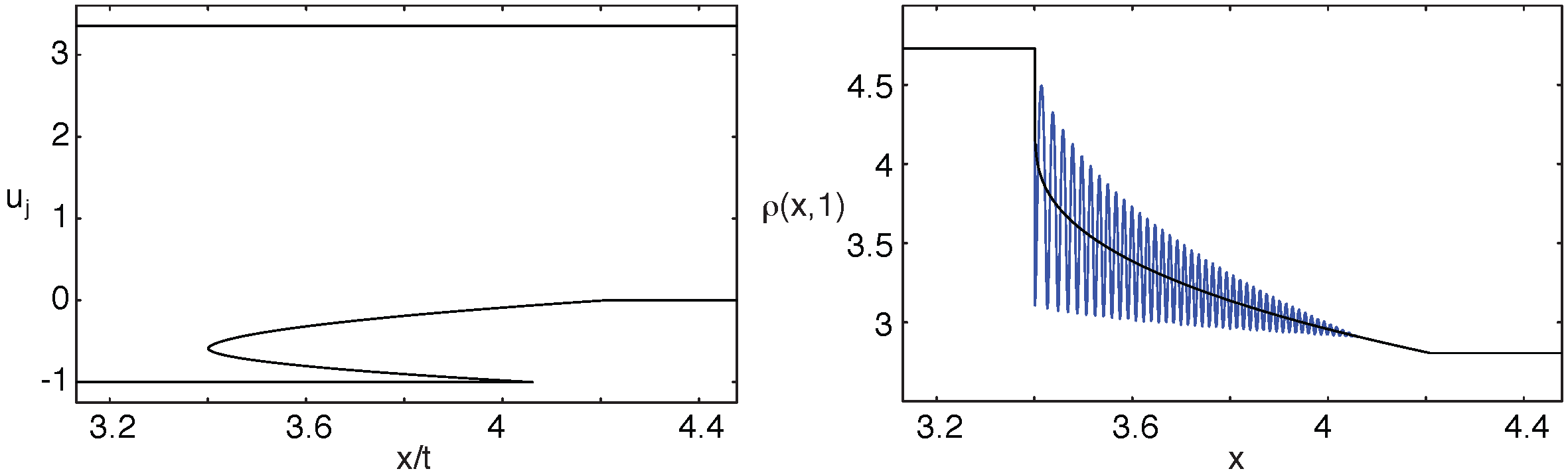}
\caption{\label{fig8} Self-Similar solution of the Whitham equations and
the corresponding oscillatory solution (\ref{periodicsolution2}) of the mKdV equation with
$\epsilon = 0.008$.   The initial data is
given by (\ref{step1}) with
  $a=4$, $b= -1$ and $c=0$ of type VIII.
  The solution in the region $4.1<x/t<4.2$
represents a rarefaction wave.}
\end{center}\end{figure}

\section{Other Initial Data}

We conclude the paper by showing how to handle the initial data (\ref{step2}).
Inequalities (\ref{34}) are replaced by
\begin{equation*}            
\frac{\pd \lambda_2 }{\pd u_2} < \frac{3}{2}\,
\frac{\lambda_1 - \lambda_2 }{u_1 - u_2} < \frac{\pd
  \lambda_1}{\pd u_2}  
\end{equation*} 
for $u_4 < u_3 < u_2 < u_1 $. Results similar to Lemma 2.2 can then be easily
proved. The rest of calculations are also similar to those in Sections 3 and 4.

\appendix
\section{\label{A} Leading Edge Calculations}
The function $I(u_1,u_2,u_3,u_4)$ of (\ref{I}) can be written in terms of the complete elliptic integral of
the first kind $K(s)$, i.e.,
$$I = {K(s) \over \sqrt{(u_1 - u_3)(u_2 - u_4)}} \ , $$
where 
\begin{equation*}
s = {(u_1 - u_2)(u_3 - u_4) \over (u_1 - u_3) (u_2 - u_4)} \  .
\end{equation*}
Using the derivative formula
$${d K(s) \over d s} = {E(s) - (1-s) K(s) \over 2s (1-s)} \ ,$$
where $E(s)$ is the complete elliptic integral of the second kind, we
calculate $\lambda_2$ and $\lambda_3$ of (\ref{lambda})
\begin{eqnarray*}
\lambda_2 &=& 2 \sigma_1- 4 \ {u_1 - u_2 \over 1 - {u_1 - u_3
\over u_2 - u_3} \ {E \over K}} \ , \\
\lambda_3 &=& 2 \sigma_1 + 4 \ {u_3 - u_4 \over 1 - {u_2 - u_4
\over u_2 - u_3} \ {E \over K}} \ .
\end{eqnarray*}
We then use (\ref{eq19}) to write $\mu_2 - \mu_3$ as
\begin{equation*}
-2 (u_2 - u_3)K(s) \left [{u_1 - u_2 \over (u_2 - u_3) K - (u_1 - u_3) E} \ {\pd q \over \pd u_2}
+ {u_3 - u_4 \over (u_2 - u_3)K - (u_2 - u_4) E} \ {\pd q \over \pd u_3} \right ] \ . 
\end{equation*}
Hence, $G(u_1,u_2,u_3,u_4)$ of (\ref{G}) becomes
\begin{equation}
\label{G2}
G=-2 \left [ {u_1 - u_2 \over (u_2 - u_3) K - (u_1 - u_3) E} \ {\pd q \over \pd u_2}
+ {u_3 - u_4 \over (u_2 - u_3)K - (u_2 - u_4) E} \ {\pd q \over \pd u_3} \right ] \ .
\end{equation}

We now use the asymptotics of $K(s)$ and $E(s)$ as $s$ is close to $1$
$$K(s)  \approx  \frac{1}{2} \log \frac{16}{1 - s}  \ , \quad 
E(s)  \approx  1 + \frac{1}{4}(1 - s)\left(\log \frac{16}{1 - s} - 1\right)$$
to calculate the $u_2=u_3$ limit
\begin{equation}
\label{G3}
G = 2 \left({\pd q \over \pd u_2} + {\pd q \over \pd u_3}\right) \ .
\end{equation}

Finally, we can also use the expression (\ref{q}) and the derivative formula
$${d E(s) \over d s} = {E(s) - K(s) \over 2 s}$$
to evaluate the partial derivatives of $G$ in the $u_2=u_3$ limit
\begin{align}
\label{G3'}
{\pd G \over \pd u_2}\Big|_{u_2=u_3} &= 2 + {(u_1 - u_4)(u_1 + 4 u_3 + u_4)
\over 4(u_1 - u_3)(u_4 - u_3)} \ ,\\
\nonumber
{\pd G \over \pd u_3}\Big|_{u_2=u_3} &= 2 - {(u_1 - u_4)(u_1 + 4 u_3 + u_4)         
\over 4(u_1 - u_3)(u_4 - u_3)} \ .      
\end{align}

\bigskip

\noindent
{\bf Acknowledgments.}
We thank Ed. Overman for showing us his numerical simulations of the NLS and
mKdV equations for small dispersion.
Y.K. and F.-R. T. were supported in part by NSF Grant DMS-0404931 and V.P. was
supported in part by NSF Grant DMS-0135308.
 
\bibliographystyle{amsplain}

\end{document}